\documentclass[journal]{IEEEtran}
\usepackage[T1]{fontenc}% optional T1 font encoding
\usepackage{graphicx}
\usepackage{cite}
\usepackage{epstopdf}
\usepackage{url}
\usepackage{amssymb, amsmath,amsthm, amsfonts}
\usepackage{ifthen}
\usepackage{tikz}
\usetikzlibrary{arrows}
\usepackage{enumerate}
\usepackage{verbatim}
\usepackage{float}
\usepackage{subcaption}

\allowdisplaybreaks

\DeclareSymbolFont{bbold}{U}{bbold}{m}{n}
\DeclareSymbolFontAlphabet{\mathbbold}{bbold}

\def\eE{\mathbb E}

\def\bhcX{\mathcal{X}^n}
\def\bhcY{\mathcal{Y}^n}
\def\bhcS{\mathcal{S}^n}

\def\sCapa{\mathcal C_{\text{est}}^{\text{mac}}(P_1,P_2)}
\def\sCapb{\mathcal C_{\text{3user}}(P_1,P_2,P_3)}
\newcommand\independent{\protect\mathpalette{\protect\independenT}{\perp}}

\newcommand{\etal}{\textit{et al.}}
\def\independenT#1#2{\mathrel{\rlap{$#1#2$}\mkern2mu{#1#2}}}

\newtheorem{theorem}{Theorem}%[section]
\newtheorem{corol}[theorem]{Corollary}
\newtheorem{lemma}[theorem]{Lemma}

\newtheorem{defn}[theorem]{Definition}
\newtheorem{remark}[theorem]{Remark}

\newcommand{\be}{\begin{equation}}
\newcommand{\ee}{\end{equation}}
\newcommand{\ben}{\begin{equation*}}
\newcommand{\een}{\end{equation*}}
\newcommand{\ba}{\begin{eqnarray}}
\newcommand{\ea}{\end{eqnarray}}
\allowdisplaybreaks
\DeclareMathOperator{\Var}{Var}

\DeclareMathOperator{\Cov}{Cov}

\newcommand{\vZ}{\sigma_Z^2}

\ifCLASSINFOpdf
  
\else

\fi

\usepackage{amsmath}
\begin{document}

\title{Joint State Estimation and Communication over a State-Dependent Gaussian Multiple Access Channel}

\author{Viswanathan~Ramachandran,
        Sibi Raj B.~Pillai,~\IEEEmembership{Member,~IEEE,}
        and~Vinod M.~Prabhakaran,~\IEEEmembership{Member,~IEEE}% <-this % stops a space
\thanks{Manuscript submitted on xxx. The work was supported in part by the Bharti Centre for
Communication, IIT Bombay and the grant 17ISROC008 from the ISRO-IITB Space Technology Cell. The material in
this paper was presented in part at the 2018 Twenty-fourth National Conference on Communications (NCC), Hyderabad, India. 
%IEEE~\cite{vrsrpvp}.
}        
\thanks{V. Ramachandran and S.R.B. Pillai are with the Department of Electrical Engineering, Indian Institute of Technology Bombay,  Mumbai, India. V.M. Prabhakaran is with the School of Technology and Computer Science, Tata Institute of Fundamental Research, Mumbai, India.}}
% <-this % stops a space

% The paper headers
%\markboth{ REPLACE THIS LINE WITH YOUR PAPER IDENTIFICATION NUMBER }%
%{ REPLACE THIS LINE WITH YOUR PAPER IDENTIFICATION NUMBER }

\maketitle

\begin{abstract}
A hybrid communication network with a common analog signal and  an
independent digital data stream as input to each node in a multiple access network
is considered. The receiver/base-station
has to estimate the analog signal with a given fidelity, and decode the
digital streams with a low error probability. Treating the analog signal
as a common state process, we set up a joint
state estimation and communication problem in a Gaussian multiple access
channel (MAC) with additive state. The transmitters have non-causal knowledge of 
the state process, and need to communicate independent data streams in addition
 to facilitating state estimation at the receiver.  We first 
provide a  complete characterization of the optimal trade-off between mean squared error distortion 
performance in estimating the state and the data rates for 
the message streams from two transmitting nodes. This is then generalized
to an $N-$sender MAC.  To this end, we show a natural connection between 
the state-dependent MAC model and a hybrid multi-sensor network in which a 
common source phenomenon is observed at $N$ transmitting nodes. Each node 
encodes the source observations as well as an independent message stream 
over a Gaussian MAC without any state process. The reciever is interested 
estimating the source and all the messages. Again the distortion-rate 
performance is characterized.
\end{abstract}


% For peer review papers, you can put extra information on the cover
% page as needed:
% \ifCLASSOPTIONpeerreview
% \begin{center} \bfseries EDICS Category: 3-BBND \end{center}
% \fi
%
% For peerreview papers, this IEEEtran command inserts a page break and
% creates the second title. It will be ignored for other modes.
\IEEEpeerreviewmaketitle

\section{Introduction}
\label{sec:intro}
Hybrid digital radio systems~\cite{wiki:xxx},\cite{wang2016simultaneous}, involving analog and digital information superposed in the same communication signal, are getting increasingly popular nowadays. In these systems, the receiver must estimate the analog signal, while also decoding the digital information. Such systems can be modelled as state-dependent channels~\cite{el2011network}, where the transmitter aids the receiver in estimating the channel state while also conveying a stream of messages~\cite{sutivong2005channel}. This is a simultaneous estimation and communication problem. For an additive channel, when the channel noise and the state are independent Gaussian processes, Sutivong \etal~\cite{sutivong2005channel} established the optimal trade-off between the mean squared error distortion in estimating the state and the communication rate in a point-to-point setting. Joint state estimation and communication is also relevant in the context of multi-user networks with several sensor nodes observing a common phenomenon. The base station/receiver is interested not only in the source process but also in the data from each node, this is the topic of this work.

%Let us briefly review state-dependent channels in communications and information theory. 
Channels with state are used to model situations in which the channel statistics are 
controlled by an external random process, known as the state process. The state process may be known 
either at the encoder, decoder, or both.
%only at the encoder, or only at the decoder or at both the terminals. 
The encoder state information can be either causal or non-causal (i.e. the entire 
state sequence is known a priori). Seminal papers by Shannon~\cite{shannon1958channels} (causal case) and Gelfand \& Pinsker~\cite{gel1980coding} (non-causal case) introduced state-dependent models. 
The latter model was motivated by coding for memory with defects, 
first studied in \cite{kuznetsov1974coding}. Costa~\cite{costa1983writing} introduced the notion 
of \emph{dirty paper coding} (DPC) for a state-dependent AWGN channel with non-causal state knowledge at the encoder, wherein the surprising conclusion that the capacity is unchanged by the presence of the state was arrived at. DPC later found extensive applications in broadcast settings, leading to the solution of the MIMO broadcast capacity region~\cite{weingarten2006capacity}. 

In certain state dependent channels, the transmitter may wish to aid the receiver in estimating the channel state, in addition to communicating messages. 
%For joint state estimation and communication in 
%provided the optimal distortion-rate trade-off In particular, the technique used involved the transmitter 
%amplifying the state by sending a scaled version of it, known as \emph{state amplification} was employed by \cite{sutivong2005channel}. 
For a point-to-point AWGN channel with additive state, splitting the available average power between 
uncoded transmission of the state and DPC for the message was found to be optimal
 for the mean squared error distortion measure~\cite{sutivong2005channel}. 
There is only limited success in extending this result to other models. For example,
%Furthermore, the message stream is communicated via dirty paper coding. 
the discrete memoryless counterpart of the state estimation problem was analyzed by Kim \etal~\cite{kim2008state}, for a restricted setting where the distortion is measured in terms of the~\textit{state uncertainty reduction rate}. Under non-causal state knowledge at the encoders, there is no known network setting where the joint state estimation and communication  trade-off is completely available, to the best of our knowledge.
The current paper solves this for the AWGN MAC setting. We now briefly mention some of the relevant contributions in the literature.

Following \cite{sutivong2005channel}, the idea of channel state amplification has been studied in several network information theoretic settings. In \cite{zhao2014capacity}, the problem of communicating a common source and two independent messages over a Gaussian broadcast channel (BC) without state was analyzed, and it was shown that a power splitting strategy to meet the two different goals is not optimal. The problem of communicating channel state information over a state dependent discrete memoryless channel with causal state information at the encoder was analyzed in \cite{choudhuri2013causal}. More recently, it was shown in \cite{bross2016conveying} that for simultaneous message and state communication over memoryless channels with memoryless states, feedback can improve the optimal trade-off region both for causal and strictly-causal encoder side information. The dual problem of \cite{sutivong2005channel}, known as \emph{state masking}, in which the transmitter tries to conceal the state from the receiver was studied in \cite{merhav2007information}. In \cite{koyluoglu2016state}, a state-dependent Gaussian BC was considered with the goal of amplifying the channel state at one of the receivers while masking it from the other receiver, with no message transmissions. \cite{liu2009message} gave inner bounds for simultaneous message transmission and state estimation over a 
state-dependent Gaussian BC. For message communication and state masking over a discrete memoryless BC, inner and outer bounds were derived in \cite{dikshtein2018broadcasting}.

The main concern of this paper is state estimation and communication in a scalar 
Gaussian multiple access setting. We will present a model with two senders first. 
%channel (GBC) with state-dependent links, 
For the model shown in Fig. \ref{fig:mac}, a common
additive independent and identically distributed (IID) Gaussian state-process affects the transmissions from both senders.
\begin{figure}[h]
\begin{center}
\includegraphics[scale=0.5]{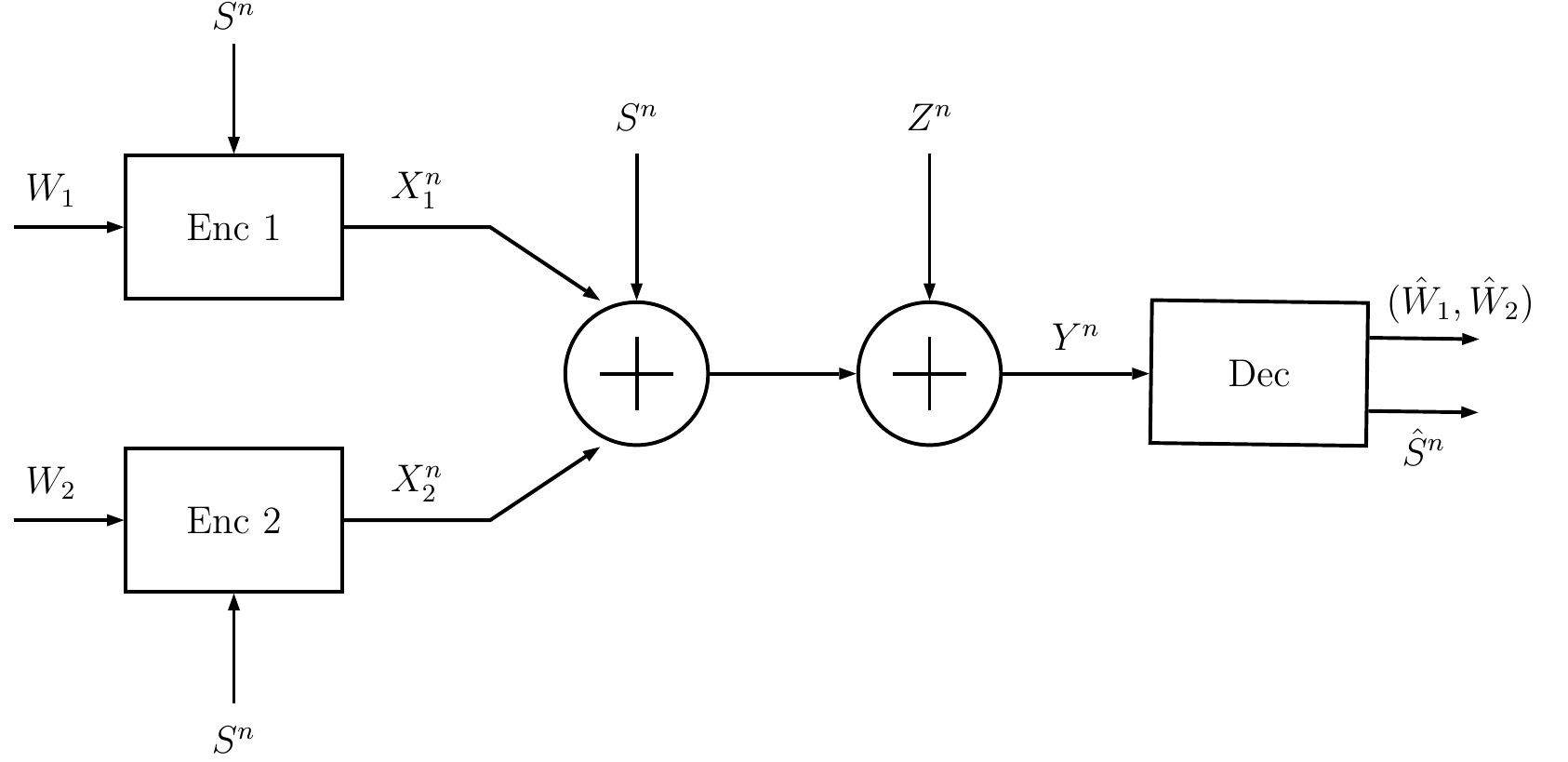}
%\begin{tikzpicture}[line width=1.0pt]
%\node (e1) at (1,1.5) [rectangle, draw, minimum height=1.0cm]{Enc~$1$}; 
%\node (e2) at (1,-0.75) [rectangle, draw, minimum height=1.0cm]{Enc~$2$}; 
%\node (c1) at (4,0) [circle, draw]{$+$};
%\node (d1) at (5,0) [rectangle, draw, right, minimum height=1.0cm]{Decoder};
%\node (s1) at (-0.75,-2.5) {$S^n$};
%
%%\foreach \i in {-0.3,0.2} {\fill (1,\i) circle(0.05cm);}
%
%\draw[->] (e1.east) -- (c1) node [midway, above, sloped] (x1) {$X_1^n$} ;
%\draw[->] (e2.east) -- (c1) node [midway, above, sloped] (x2) {$X_2^n$};
%\draw [<-] (c1) --++(0,0.75) node[above]{$Z^n$};
%\draw[->] (c1) -- (d1);
%
%\draw[<-] (e1) --++(-1.75,0) node[left]{$W_1$};
%\draw[<-] (e2) --++(-1.75,0) node[left]{$W_2$};
%
%
%\draw[->] (s1)  -| (e2);
%\draw[->] (s1) ++(0.75,0)--++(0,3.25) -| (e1);
%\draw[->] (s1) --++(2.5,0) --++(0,0) -| (c1);
%\draw[->] (d1.30) --++(0.3,0) node[right,scale=0.8]{$\hat W_1$};
%\draw[->] (d1.5) --++(0.3,0) node[right,scale=0.8]{$\hat W_2$};
%\draw[->] (d1.-20) --++(0.3,0) node[right,scale=0.8]{${\hat S}^n$};
%
%\end{tikzpicture}
\end{center}
\caption{Dirty Paper MAC with State Reconstructions\label{fig:mac}}
\end{figure}
The transmitters know the state-process in a non-causal fashion.
Our first objective is to obtain an estimate of the state process at the receiver to within a prescribed distortion bound. 
In addition, there is a message stream from each encoder to the receiver. Given a rate pair for their respective private messages, the transmitters attempt to minimize the distortion incurred in state estimation at the receiver. Under individual average transmit power constraints at the encoders, the 
Gaussian state dependent MAC with state estimation requirement leads to interesting trade-offs 
between the achievable distortion and the rates. 
We name this model as the \emph{dirty paper MAC with state estimation}. 

In the absence of a state process, the AWGN MAC capacity  can be seen as a natural extension
of the point to point AWGN model~\cite{cover2012elements}. When the state is present, it 
is tempting to look for such an extension of the joint state estimation and communication tradeoff,
using the single user  results in \cite{sutivong2005channel}. However, notice that the former
connection is greatly aided by the polymatroidal capacity region of a Gaussian MAC (GMAC). Essentially
three inequalities suffice to establish the converse result for a two user MAC~\cite{cover2012elements}.
On the other hand, for joint estimation and communication in a two user MAC, 
even the cross-section of the optimal rate-region under a given distortion is not always a polytope.
This explains why single user  techniques are not enough in our setup.
%Notice that our model extends  the state estimation problem in 
%\cite{sutivong2005channel} to the multiple access channel. 
Our main contributions are summarized below.
\begin{itemize}
\item We provide a complete characterization of the optimal trade-off between joint state estimation and 
communication over a two user dirty paper MAC with state estimation. 
%Notice that even for a given distortion, the rate region is not always a polymatroid, unlike the MAC without state~\cite{cover2012elements}. Thus new techniques were required to demonstrate a converse to the capacity region.
\item For a multi-sensor network of nodes observing a common source phenomenon, with each node 
possibly having an additional independent message stream, we characterize the optimal distortion-rate 
performance in joint state-estimation and communication over a Gaussian MAC without state. 
The model is sufficiently general to  include cases where the source symbols also act as 
additive state, which is known non-causally at the transmitters. 
%
%with two of the nodes
%having additional data streams to communicate to a fusion-center, we show that uncoded transmission
%of the state by nodes not having any data is optimal for joint source estimation (under squared 
%distortion) and message communication. 
%
%\item Extension of the coding scheme to a general $N$ sender Gaussian MAC with message and 
%source communication requirements is also presented.
\end{itemize}
%We note that the conference version~\cite{vrsrpvp} of the current paper dealt exclusively with the sum-capacity versus distortion trade-off. Also, 
The transmission of correlated sources through a MAC is a very important open problem in literature~\cite{coversalehi80,dueck81}. Our problem is related, but an extreme case called
the cooperative MAC~\cite{coversalehi80}, where the source observation is common to all the transmitters.
In our model, we are only constraining the reconstruction fidelity, whereas \cite{coversalehi80} considers
the lossless case. In some sense, the problem which comes closest to the one here is the source estimation 
problem in \cite{gastpar08}. Here $N$ transmitters in a Gaussian MAC observe independent noisy versions 
of a single source. The transmitters are assumed to be \emph{symmetric}, i.e. they
have identical power constraints, and the same noise variance in the source observations. 
The channel state process is completely absent, but uncoded transmissions
turn out optimal. Some relaxations on symmetric users are provided in 
\cite{tian2015,tian2017matched}. In fact \cite{tian2015} also considers the scalar single user model with 
additive state non-causally known at the transmitter, and shows that uncoded transmission is optimal 
for state estimation in the Gaussian setting. However, even for the point to point system, optimal communication schemes in
presence of additional messages are  unknown when the source observations are noisy~\cite{tian2015}. Other relevant studies regarding communication of sources over a MAC include \cite{xiao2007multiterminal} (distributed correlated sources over an orthogonal MAC), \cite{lapidoth2010sending} (bivariate Gaussian source over GMAC with individual distortion constraints), \cite{nazer2007computation} (reliable function computation over MAC) and \cite{soundararajan2012communicating} (linear functions of correlated sources over GMAC).
Notice that none of these models consider additional data streams along with source communications over MACs. For noiseless source observations, \cite{sutivong2005channel} characterizes the optimal
tradeoff for message as well as rate, the current paper extends this to a $N$-sender GMAC, with or without state.

%Notice however that in the presence of a data-stream, even the single user
%version of \cite{gastpar08} is unsolved~\cite{tian15}.  
%In our MAC model, each node has a private data stream, however the source observations are not 
%corrupted by noise, making the model more tractable.

%In particular, we derive a complete characterization of the optimal rates versus distortion trade-off. We note that the converse arguments from the single-user case of Sutivong \etal~\cite{sutivong2005channel} can be extended in a relatively straightforward manner to obtain the sum-capacity versus distortion trade-off in the current setting. However, additional arguments are necessary to characterize the trade-off between the individual rates and distortion, which is the novelty of the current paper. We would like to point out that the conference version~\cite{vrsrpvp} of the current paper dealt exclusively with the sum-capacity versus distortion trade-off.
%In addition, we also show that for state estimation over an additive Gaussian noise channel with arbitrary additive state (not necessarily Gaussian), the optimal trade-off between rate and state estimation distortion occurs when the additive state is Gaussian.

\textit{Notations:} We use $\mathbb{P}(\cdot)$ to denote the probability of an event, and $\mathbb{E}[\cdot]$ to denote the expected value of a random variable. All logarithms in this paper are to the base $2$, unless specified otherwise. We denote random vectors as $U^n:= U_1, \cdots, U_n$ and $U_1^n:= U_{11},\cdots, U_{1n}$. Calligraphic letters represent the alphabets. $\|.\|$ denotes the Euclidean norm of a vector.

The paper is organized as follows: we introduce the system model and main results
in Section~\ref{sec:sys}. Sections~\ref{sec:achieve} and \ref{sec:conv} respectively contain the achievable coding scheme  and converse to the optimal region. Section \ref{sec:Nuser} considers the generalization to $N$ transmitters over a  GMAC, with and without state. Concluding remarks are given in Section~\ref{sec:concl}.

\section{System Model and Results}\label{sec:sys}
The dirty paper MAC with state estimation is shown in Fig.~\ref{fig:mac}. Here $S \sim \mathcal{N}(0,Q)$ is the channel state and $Z$ is the channel noise, with $S \independent Z$. The state and noise processes are i.i.d., with the state being non-causally available at both encoders. The receiver observes (for a single channel use)
\begin{equation}
Y = X_1+X_2 + S + Z, \label{eq:syst}
\end{equation}
where $Z \sim \mathcal N(0,\vZ)$. After $n$ observations, the decoder  estimates 
$\hat{S}^n = \phi(Y^n)$ using
a reconstruction map $\phi(\cdot):\mathcal Y^n \rightarrow \mathbb R^n$, 
and also decodes the independent messages $(W_1, W_2)$, which are assumed be independent of $S^n$.
We also take  $W_j$ to be uniformly drawn from $\{1,\cdots,2^{nR_j}\}$ for  $j=1,2$.

Our objective here is to maintain the distortion below a prescribed value, 
while ensuring that the average error probability of decoding the messages is  small enough, i.e. %~\cite{cover2012elements}. 
\begin{align}
\frac{\eE{[\|S^n-\phi(Y^n)\|^2]}}{n} &\leq D+\epsilon, \label{eq:dist:const} \\ 
\mathbb{P}(\psi(Y^n) \neq (W_1,W_2)) & \leq \epsilon  \label{eq:error:const}.
\end{align}
Here $\psi:\bhcY \to \{1,\cdots, 2^{nR_1}\} \times \{1,\cdots, 2^{nR_2}\}$ is the decoding map, $D$ represents the distortion target, and $\epsilon > 0$ the probability of error target.

\begin{defn}
A scheme achieving \eqref{eq:dist:const} -- \eqref{eq:error:const} using 
the encoder maps $\mathcal{E}_j:\{1,\cdots,2^{nR_j}\} \times \bhcS \to \bhcX_j$ 
such that $\eE \|X_j^n\|^{2} \leq nP_j$, $j=1,2$,
along with two maps $\phi(\cdot)$ and $\psi(\cdot)$ at the receiver is 
called an $(n,R_1,R_2, D,\epsilon)$ communication scheme. 
%
%An $(n,R_1,R_2,D,\epsilon)$ communication scheme consists of 
%two encoder maps $\mathcal{E}_j:\{1,\cdots,2^{nR_j}\} \times \bhcS
%\to \bhcX_j$, $j=1,2$, 
%a decoding map $\psi:\bhcY \to \{1,\cdots, 2^{nR_1}\} \times \{1,\cdots, 2^{nR_2}\}$, % \times [1:2^{n \hat R_i}]$,
%and a receiver reconstruction map $\phi: \bhcY \to \mathbb R^n$ such that for $W_i$ uniformly distributed over $\{1,\cdots,2^{nR_i}\}$, $W_1 \independent W_2$, and $X_j^n=\mathcal{E}_j(W_j,S^n)$, we have 
%\begin{equation}
%\frac{\eE{[\|S^n-\phi(Y^n)\|^2]}}{n} \leq D+\epsilon, 
%\end{equation}
%\begin{equation}
%\mathbb{P}(\psi(Y^n) \neq (W_1,W_2)) \leq \epsilon,
%\end{equation}
%under average power constraints $\eE \|X_j^n\|^{2} \leq nP_j$, $j=1,2$. 
\end{defn}

We say that a triple $(R_1,R_2,D)$ is achievable if a $(n,R_1,R_2,D,\epsilon)$
communication scheme exists for every $\epsilon > 0$, possibly by taking $n$ large enough. Let 
$\sCapa$ be the closure of the set of all achievable $(R_1,R_2,D)$ triples, with
$0 \leq D \leq Q$. Our main result is stated below.
\begin{theorem} \label{thm:main}
For the dirty-paper MAC with state, the optimal trade-off region $\sCapa$ is given by the convex closure of all $(R_1,R_2,D) \in \mathbb{R}_{+}^{3}$ such that
\begin{gather}
R_1 \leq \frac{1}{2}\log\left(1+\frac{\gamma P_1}{\vZ}\right), \label{eq:rc1}\\
R_2 \leq \frac{1}{2}\log\left(1+\frac{\beta P_2}{\vZ}\right), \label{eq:rc2}\\
R_1+R_2 \leq \frac{1}{2}\log\left(1+\frac{\gamma P_1+\beta P_2}{\vZ}\right), \label{eq:rc3}\\
D \geq \frac{Q(\vZ\!+\!\gamma P_1\!+\!\beta P_2)}{P_1\!+\!P_2\!+\!Q\!+\!\vZ\!+\!2\sqrt{\bar{\gamma}P_1 Q}\!+\!2\sqrt{\bar{\beta}P_2 Q}\!+\!2\sqrt{\bar{\gamma}\bar{\beta}P_1 P_2}}. \label{eq:distc}
\end{gather}
for some $\gamma \in [0,1]$ and $\beta \in [0,1]$, with $\bar{\gamma}=1-\gamma$ and $\bar{\beta}=1-\beta$.  
\end{theorem}

Before we prove this result, notice that the rate region is not in general a polytope for any
given distortion value, unlike the  case where state estimation is not required~\cite{kim2004multiple}. 
Nevertheless, the region admits a compact representation as given in \eqref{eq:rc1}--\eqref{eq:distc}.

\begin{proof}
In Section~\ref{sec:achieve}, we present a communication scheme to achieve tuples  satisfying 
the constraints \eqref{eq:rc1}--\eqref{eq:distc}. Then in Section \ref{sec:conv}, we show a 
converse result which bounds the distortion-rate performance for any successful communication scheme.
We further show that the tradeoff cannot be better than the ones defined by 
\eqref{eq:rc1}--\eqref{eq:distc} for some values of $\gamma, \beta \in [0,1]$, this
is given in Section~\ref{sec:equivalence}.
The main novelty of the proof is in the converse result.
\end{proof}

%\sibi{Shall we remark that the region is convex without the convex closure}

\begin{remark}
Notice that on setting $\gamma=\beta=1$, which amounts to no state estimation requirements, 
%we recover the (interference-free) Gaussian MAC capacity region. This is in line with the 
we recover the multiuser writing on dirty paper result in \cite{kim2004multiple} which proves that the capacity region of the \emph{dirty} paper MAC  with state is not 
affected by the presence of state.
\end{remark}

One of the motivations of our model comes from sensor networks employed in on-board platforms. We now
make this connection more explicit. 

\subsection{Connection to Multi-sensor models}\label{sec:3user:connect}
%We illustrate the equivalence of the multi-sensor network model with message and source transmissions to our state-dependent model, by 
Let us introduce a more general $N-$sender Gaussian MAC framework as in Figure~\ref{fig:N:send},
where the receiver observes
\begin{align}
Y = \sum_{i=1}^N X_i + \alpha S + Z.
\end{align}
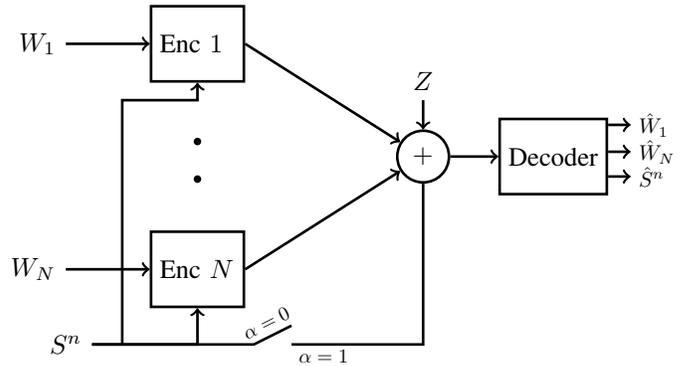
\begin{figure}[h]
\centering
\begin{tikzpicture}[line width=1.0pt]
\node (e1) at (1,1.5) [rectangle, draw, minimum height=1.0cm, text width=1.0cm]{Enc~$1$}; 
\node (e2) at (1,-1.5) [rectangle, draw, minimum height=1.0cm, text width=1.0cm]{Enc~$N$}; 
\node (c1) at (4,0) [circle, draw]{$+$};
\node (d1) at (5,0) [rectangle, draw, right, minimum height=1.0cm]{Decoder};
\node (s1) at (-0.75,-2.5) {$S^n$};

\foreach \i in {-0.3,0.2} {\fill (1,\i) circle(0.05cm);}

\draw[->] (e1.east) -- (c1);
\draw[->] (e2.east) -- (c1);
\draw [<-] (c1) --++(0,0.75) node[above]{$Z$};
\draw[->] (c1) -- (d1);

\draw[<-] (e1) --++(-1.75,0) node[left]{$W_1$};
\draw[<-] (e2) --++(-1.75,0) node[left]{$W_N$};

\draw[->] (s1)  -| (e2);
\draw[->] (s1) ++(0.75,0)--++(0,3.25) -| (e1);
\draw (s1) --++(2.5,0) --node[sloped, above,scale=0.75]{$\alpha=0$}++(0.5,0.25) ++(0,-0.25)
node[yshift=-0.15cm,right, scale=0.75]{$\alpha=1$} -| (c1);
\draw[->] (d1.30) --++(0.3,0) node[right,scale=0.8]{$\hat W_1$};
\draw[->] (d1.5) --++(0.3,0) node[right,scale=0.8]{$\hat W_N$};
\draw[->] (d1.-20) --++(0.3,0) node[right,scale=0.8]{${\hat S}^n$};

\end{tikzpicture}

\caption{Source and Message through a $N$ Sender MAC\label{fig:N:send}}
\end{figure}
When the parameter $\alpha=1$ we recover the state dependent model, and $\alpha=0$ corresponds
to a source estimation and message communication problem over a GMAC without state. 
The transmissions $X_i$ of user~$i$ is subjected to an average power constraint
of $P_i$.
Each transmitter observes the source process $S$, which is assumed to be 
non-causally available to them. The receiver should estimate the source, as well as an
independent message stream from  each transmitter. We term the model
as the \emph{source-message communication} problem. 
Notice that some of the nodes may not have any messages, they simply help in the estimation
of source. It turns out that having uncoded transmission at each node devoid of any 
messages is indeed the optimal strategy in such set ups, marking
the importance of the results presented for state-dependent models in 
the previous sub-section.  A brief literature review on  source and message 
communication is in order.

Goblick~\cite{goblick1965theoretical} showed that for transmission of Gaussian sources over 
Gaussian channels, an uncoded strategy of sending a scaled version of the source to meet the 
power constraint and then MMSE estimation at the receiver is optimal. The same approach can 
be used to communicate Gaussian sources over Gaussian broadcast channels, as was observed in 
Prabhakaran et. al~\cite{prabhakaran2011hybrid}. We already mentioned that
the  uncoded approach is optimal for source estimation in some symmetric GMAC models, 
in the absence of messages~\cite{gastpar08}. In the presence of independent private 
message stream at each node in an $N-$sender  Gaussian MAC, we show that 
uncoded source transmissions are indeed optimal while communicating a common source observation.

%\sibi{Are there more such models?}

Our main result for the $N-$sender Gaussian MAC with source-message communication 
is stated in the following theorem.
\begin{theorem} \label{thm:main:2}
For an $N$-sender GMAC with message and state communication, the optimal trade-off region is given by the convex closure of the set of $(R_1,R_2,\cdots,R_{N},D)$ such that
\begin{gather}
\sum_{j \in \mathcal{J}} R_j \leq \frac{1}{2} \log \left(1+\frac{\sum_{j \in \mathcal{J}} \gamma_j P_j}{\vZ}\right) \: \forall \: \mathcal{J} \subseteq [1:N], \\
D \geq \frac{Q\left(\vZ+\sum_{j=1}^N \gamma_j P_j \right)}{\vZ+\sum_{j=1}^N \gamma_j P_j+\left(\alpha\sqrt{Q}+\sum_{j=1}^N \sqrt{(1-\gamma_j)P_j}\right)^2}, \label{eq:distc:N}
\end{gather}
for some $(\gamma_1, \cdots, \gamma_N) \in [0,1]^N$.
\end{theorem}

\begin{IEEEproof}
This is given in Section~\ref{sec:Nuser}.
\end{IEEEproof}

Let us specialize our results to a three sender Gaussian MAC, as shown in 
Fig.~\ref{fig:3enc}. 
%Here all three  encoders observe a common Gaussian source process.
% while each of the first two users also has an independent message stream. 
The channel model is
\begin{align}
Y = X_1+X_2+X_3+Z,
\end{align}
with $Z\sim \mathcal N(0,\sigma_z^2)$, and the power constraints $\eE|X_j|^2 \leq P_j, j=1,2,3$.
\begin{figure}[h]
\centering
\includegraphics[scale=0.6]{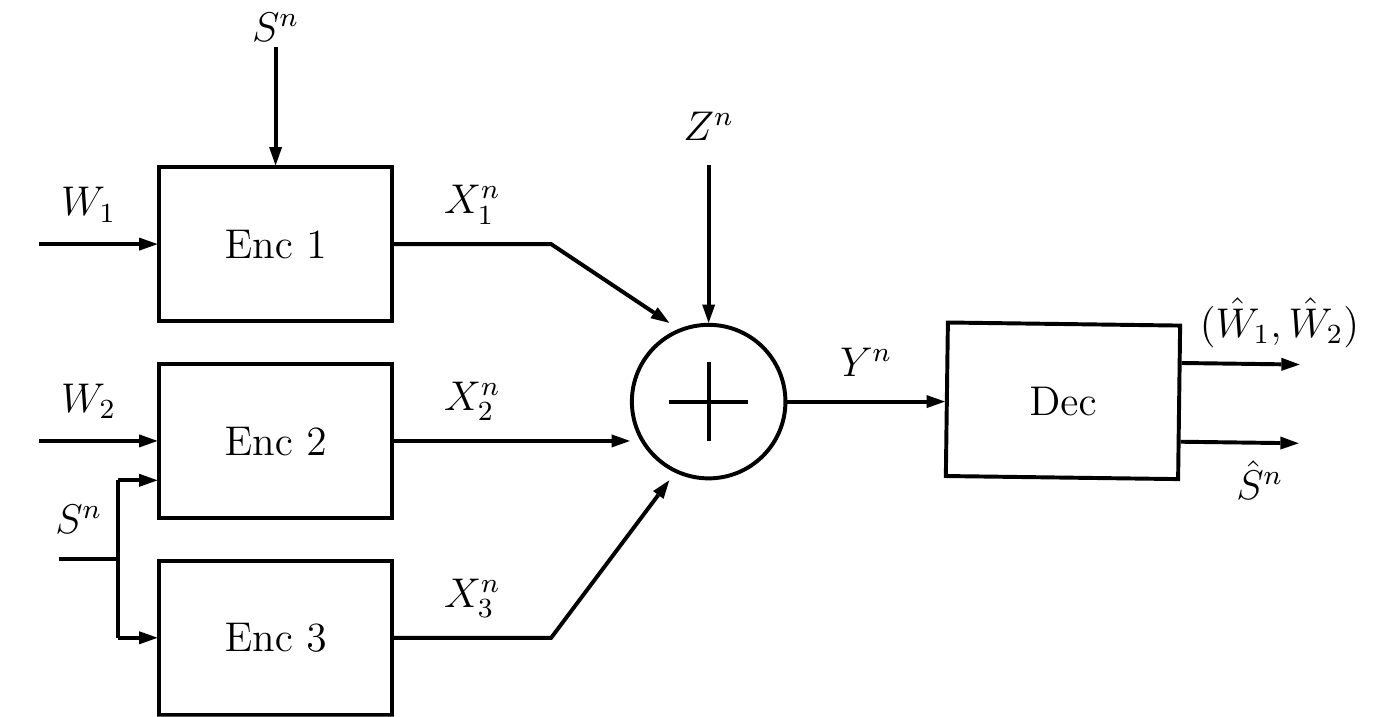}
\caption{Three user MAC with common $S^n$ and two messages\label{fig:3enc}}
\end{figure}
Suppose the third terminal is only interested in conveying the source process under a distortion constraint, 
and it follows an uncoded strategy by sending
\begin{align}
X_3 = \sqrt{\frac{P_3}{Q}} S.
\end{align}
Then the overall model becomes
\begin{align}
Y &= X_1+X_2+\sqrt{\frac{P_3}{Q}} S+Z \\
&= X_1+X_2+S^{\prime}+Z,
\end{align}
where $S^{\prime}=\sqrt{\frac{P_3}{Q}} S$ is non-causally known to both the encoders.
%We call this problem the  three sender dirty paper MAC, 
%which under uncoded transmission at the third user 
Notice that this is indeed the
 joint state estimation and communication model for a dirty paper MAC with state estimation. 
%
%By showing that  uncoded transmission is the optimal action at the third encoder, 
%clearly Theorem~\ref{thm:main} finds crucial application in such networks.
%we capture this in the following theorem.
The following corollary is of interest when the third user has no message.
\begin{corol} \label{cor:main}
For the three sender MAC with source and message communication, the optimal trade-off 
region $\sCapb$ when $\alpha=R_3=0$ is given by the convex closure of all 
$(R_1,R_2,D) \in \mathbb{R}_{+}^{3}$ such that \eqref{eq:rc1}--\eqref{eq:rc3} holds and
\begin{gather}
%R_1 \leq \frac{1}{2}\log\left(1+\frac{\gamma P_1}{\vZ}\right), \label{rc1b}\\
%R_2 \leq \frac{1}{2}\log\left(1+\frac{\beta P_2}{\vZ}\right), \label{rc2b}\\
%R_1+R_2=R_{sum} \leq \frac{1}{2}\log\left(1+\frac{\gamma P_1+\beta P_2}{\vZ}\right), \label{rc3b}\\
D \geq \frac{Q(\vZ+\gamma P_1+\beta P_2)}{\sum_{i=1}^3 P_i+\!\vZ+2\sqrt{\bar{\gamma} P_1 P_3}+2\sqrt{\bar{\beta} P_2 P_3}+2\sqrt{\bar{\gamma}\bar{\beta}P_1 P_2}}. \label{distcb}
\end{gather}
for some $\gamma, \beta \in [0,1]$ 
%and $\beta \in [0,1]$, 
with $\bar{\gamma}=1-\gamma$ and $\bar{\beta}=1-\beta$.  
\end{corol}
%\begin{proof}
%The proof is given in Section \ref{sec:ext}.
%\end{proof}

% ACHIEVABILITY
\section{Achievability for the Dirty Paper MAC}\label{sec:achieve}
We employ a suitable power splitting strategy along with dirty paper coding to prove the 
achievability. This is rather straightforward, but the details are given for completeness. 
The available power $P_1$ at encoder $1$ is split into two parts: namely $\gamma P_1$ for 
message transmission and $\bar{\gamma}P_1$ for state amplification, for some $\gamma \in [0,1]$. Likewise, the power $P_2$ available at the second encoder is split into $\beta P_2$ (message transmission) and $\bar{\beta}P_2$ (state amplification) for some $\beta \in [0,1]$. Then generate the state amplification signals
\begin{equation}
X_{1 \textrm{s} j} = \sqrt{\frac{\bar{\gamma}P_1}{Q}} S_j \:\: \textup{and} \:\: X_{2 \textrm{s} j} = \sqrt{\frac{\bar{\beta}P_2}{Q}} S_j, \, 1\leq j \leq n
\end{equation}
at the respective encoders. Now the system model in \eqref{eq:syst} can be rewritten as
\begin{align}
Y &= X_{1 \textrm{m}}+X_{1 \textrm{s}}+X_{2 \textrm{m}}+X_{2 \textrm{s}}+S+Z \notag\\
&= X_{1 \textrm{m}}+X_{2 \textrm{m}}+\left(1+\sqrt{\frac{\bar{\gamma}P_1}{Q}}+\sqrt{\frac{\bar{\beta}P_2}{Q}}\right)S+Z.
\end{align}
Here the index $\textrm{m}$ in the subscript indicates that 
the corresponding signals are intended for message transmission, while the subscript 
$\textrm{s}$ indicates state amplification signals. Now in order to communicate the messages across to the receiver, we employ the writing on dirty paper result for a Gaussian MAC~\cite{kim2004multiple}.

Recall that a known \emph{dirt} over an AWGN channel can be completely cancelled by dirty paper coding~\cite{costa1983writing}. More generally, a rate $R$ satisfying
\begin{equation}
R \leq I(U_1;Y)-I(U_1;S),
\end{equation}
when evaluated for some feasible distribution $p(u_1,s,x)p(y|x,s)$,  
can be achieved by Gelfand-Pinsker coding~\cite{gel1980coding} for a point-to-point channel 
with noncausally known state. 
In order to achieve \eqref{eq:rc1} -- \eqref{eq:rc3}, we first consider a dirty paper channel with input $X_{1 \textrm{m}}$, known state \smash{$S^{\prime}=\left(1+\sqrt{\bar{\gamma}P_1/Q}+\sqrt{\bar{\beta}P_2/Q}\right)S$} and unknown noise $X_{2 \textrm{m}}+Z$. We choose $U_1=X_{1 \textrm{m}}+\alpha_1 S^{\prime}$, $X_{1 \textrm{m}} \independent S$ with $X_{1 \textrm{m}} \sim \mathcal{N}(0,\gamma P_1)$ and \smash{$\alpha_1 = \frac{\gamma P_1}{\gamma P_1+\beta P_2+\vZ}$}. The achievable rate is
\begin{align}
R_1 &= \frac{1}{2} \log \left(1+\frac{\gamma P_1}{\beta P_2+\vZ}\right).
\end{align}
Once the $U_1^n$ codeword is decoded, it can be subtracted from $Y^n$ to obtain
\begin{equation}
\tilde{Y}^n = Y^n-U_1^n=X_{2 \textrm{m}}^n+(1-\alpha_1){S^{\prime}}^n+Z^n.
\end{equation}
Now for sender $2$, this can be considered as another dirty paper channel with input $X_{2 \textrm{m}}$, known state $S^{\prime\prime}=(1-\alpha_1)S^{\prime}$ and unknown noise $Z$. Let us choose $U_2=X_{2 \textrm{m}}+\alpha_2 S^{\prime\prime}$, $X_{2 \textrm{m}} \independent S$ with $X_{2 \textrm{m}} \sim \mathcal{N}(0,\beta P_2)$ and \smash{$\alpha_2 = \frac{\beta P_2}{\beta P_2+\vZ}$}. The achievable rate becomes
\begin{align}
R_2 &= \frac{1}{2} \log \left(1+\frac{\beta P_2}{\vZ}\right).
\end{align}
By reversing the decoding order, we can show that the following rate pair is also achievable
\begin{equation}
(R_1,R_2) = \left(\frac{1}{2} \log \left(1+\frac{\gamma P_1}{\vZ}\right), \frac{1}{2} \log \left(1+\frac{\beta P_2}{\gamma P_1+\vZ}\right)\right).
\end{equation}
The entire $(R_1,R_2)$ rate region as in expressions \eqref{eq:rc1} through \eqref{eq:rc3} can now be achieved by time sharing.

Now we turn to the proof of the achievable distortion. Based on the observation $Y^n$, the receiver forms the linear estimate
\begin{equation*}
\hat{S}^n = \frac{(Q+\sqrt{\gamma P_1 Q}+\sqrt{\beta P_2 Q})\:\: Y^n}{P_1\!+\!P_2\!+\!Q\!+\!\vZ\!+\!2\sqrt{\bar{\gamma}P_1 Q}\!+\!2\sqrt{\bar{\beta}P_2 Q}\!+\!2\sqrt{\bar{\gamma}\bar{\beta}P_1 P_2}}.
\end{equation*}
The MMSE can be readily calculated to be the RHS of expression \eqref{eq:distc}. This completes the proof of achievability.

\def\neps{n \epsilon_n}
\section{Outer Bound for the Dirty Paper MAC with State Estimation}\label{sec:conv}
In this section and the next, we show that any successful communication scheme has to satisfy the rate and distortion
constraints of Theorem~\ref{thm:main}. Two ideas from the single user result of 
\cite{sutivong2005channel} will turn out to be very useful towards the proof. The first is stated below as a 
lemma, its proof can be found in \cite{sutivong2005channel}. 
\begin{lemma} \label{lem:dist:MI}
Any communication scheme achieving a distortion 
$D_n \doteq \frac{1}{n} \eE{||S^n-\hat{S}^n||^2}$ over  block length $n$ will have
%The following inequality holds:
\begin{equation} 
\frac{n}{2} \log\left(\frac{Q}{D_n}\right) \leq I(S^n;Y^n). \label{eq:dist:MI}
\end{equation}
\end{lemma}
The second useful idea is to construct bounds for the term $R+ \lambda \log \frac Q{D_n}$, instead of 
separate bounds for rate $R$ and distortion $D_n$. As in the single user case, the
above transformation of the distortion function turns out to be sufficient for the dirty paper MAC 
with state as well, however we now have to consider rate-pairs $(R_1,R_2)$. In addition, we will use the following property of Gaussian random variables~\cite{cover2012elements}:
%\begin{itemize}
Gaussians maximize entropy, i.e. for  $X_g^n \sim \mathcal N(0,K)$
\begin{align} \label{eq:gauss:ent}
h(X^n) \leq h(X_g^n) \text{ when ever } \Cov(X^n) \preceq K.
\end{align}
%\item Conditioning reduces entropy.
%\end{itemize}
The above facts will be extensively used in our proofs.
For $(\mu_1, \mu_2, \lambda) \in { {\mathbb R}^+}^3$, let us define
 $$
T(\mu_1, \mu_2, \lambda) = \max \mu_1 R_1 + \mu_2 R_2 + \frac\lambda 2 \log \frac QD,
$$
where the maximum is over all $(R_1, R_2, D)$ obeying \eqref{eq:rc1} -- \eqref{eq:distc}. 
Notice that we did not
consider $\mu_i < 0$, as this will trivially correspond to $R_i=0$ in
the maximization, a case already accounted for by $\mu_i=0$. Similarly, since $D\leq Q$,
we need to consider only $\lambda \geq 0$.  Thus, only non-negative weighing coefficients
are considered in the sequel. A converse proof can be obtained by showing that
if $(R_1,R_2,D_n)$ is achievable using block length $n$, then, for all $\mu_1,\mu_2,\lambda \geq 0$,
\begin{align} \label{eq:up:bnd:1}
\mu_1 R_1 + \mu_2 R_2 + \frac \lambda 2 \log \frac Q{D_n} \leq T(\mu_1, \mu_2, \lambda)+o(1).
\end{align}
Our strategy is to convert the LHS of \eqref{eq:up:bnd:1} to a form where \eqref{eq:gauss:ent} can be
applied. Since the messages $(W_1,W_2)$ are independent of $S^n$, we have the
Markov condition $X_1^n \rightarrow S^n \rightarrow X_2^n$. Denoting
\begin{align*}
\Var{[X|Y]} &\triangleq \min_{\alpha} \eE{[X-\alpha Y]^2},
\end{align*}
we have for the $i$-th entry in a block,
\begin{align}
\Var{[X_{1i} + X_{2i}|S^n]} &= \Var{[X_{1i}|S^n]}+\Var{[X_{2i}|S^n]}  \notag \\
	&\leq \Var{[X_{1i}|S_i]}+\Var{[X_{2i}|S_i]}.
\end{align}
%To bound the weighted sum $\mu R_1+R_2+\frac{\lambda}{2} \log\left(\frac{Q}{D}\right)$ above, we
More generally, we can define the \emph{empirical} covariance matrix $K_i$ of the combined vector
%first introduce certain parameters based on the empirical covariance matrix of the random vector 
$(X_{1i},X_{2i},S_i)$,  %$(X_{1i},X_{2i},S_i,Z_i,Y_i)$. 
%For every $i \in [1,\cdots,n]$, consider the random vector $(X_{1i},X_{2i},S_i,Z_i,Y_i)$ and denote its covariance matrix by $K_i$. 
with $K_i(m,n)$ denoting its entries.  Let us denote
\begin{align*}
K_i(1,1) &= \eE |X_{1i}|^2 = P_{1i} \\
K_i(2,2) &= \eE |X_{2i}|^2 = P_{2i}.
\end{align*}
Now, let us introduce two parameters  $\gamma_i, \beta_i \in [0,1]$ for each $i\in \{1,\cdots, n\}$ 
such that 
\begin{align} 
\begin{split} \label{eq:ki:entry}
K_i(1,3) &= \eE{[X_{1i} S_i]} = \eta_{1i} \sqrt{(1-\gamma_i)P_{1i}Q} \\
K_i(2,3) &= \eE{[X_{2i} S_i]} = \eta_{2i} \sqrt{(1-\beta_i)P_{2i}Q}, 
\end{split}
\end{align}
where $\eta_{ji} \in \{-1,+1\}, j=1,2$ is the sign of  the correlation. 
The remaining terms of $K_i$ can be evaluated using \eqref{eq:ki:entry} and
\begin{align*}
\Var{[X_{1i}|S_i]} &= \min_{a}\eE{[X_{1i}-a S_i]^2}=\gamma_i P_{1i} \\
\Var{[X_{2i}|S_i]} &= \min_{a}\eE{[X_{2i}-a S_i]^2}=\beta_i P_{2i}.
\end{align*}
Let us also define two parameters $\gamma\in[0,1]$ and $\beta \in [0,1]$:
\begin{align} \label{eq:gam:bet}
\gamma = \frac 1{nP_1} \sum_{i=1}^n \gamma_{i}P_{1i} \text{ and }
\beta = \frac  1{nP_2} \sum_{i=1}^n \beta_{i}P_{2i}.
\end{align}
With this, we are all set to  prove \eqref{eq:up:bnd:1}. First of all, considering
$\mu_1 \geq \mu_2$ is sufficient, as a simple renaming of the indices will 
give us the opposite case. For $\mu_2>0$, since $\lambda$ is an arbitrary positive number, 
we can equivalently maximize $\mu_1 R_1 + \mu_2 R_2 + \mu_2 \lambda \frac 12 \log \frac Q{D_n}$.
Dividing by $\mu_2$, and then renaming $\frac{\mu_1}{\mu_2}$ as $\mu$, the maximization
becomes $\forall \mu \geq 1, \lambda \geq 0$,
\begin{align} \label{eq:mu:opt}
\max \, \mu R_1 + R_2 + \frac \lambda 2 \log \frac Q{D_n}.
\end{align}

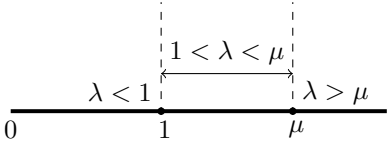
\begin{figure}
\begin{center}
\begin{tikzpicture}
\draw[line width=1.5pt] (0,0) --++(5,0);
\coordinate (c1) at (2,0);
\coordinate (c2) at (3.75,0);

\node at (0,0) [below] {$0$};
\fill (c1) circle (0.05cm) node[below, xshift=0.05cm]{$1$};
\fill (c2) circle (0.05cm) node[below, xshift=0.05cm]{$\mu$};

\draw[dashed] (c1) ++(0,-0.05) --++(0,1.5);
\draw[dashed] (c2) ++(0,-0.05) --++(0,1.5);

\draw (c1) ++(0,0.25) node[left]{$\lambda < 1$};
\draw (c2) ++(0,0.25) node[right]{$\lambda > \mu$};

\draw[<->] (2,0.5) --node[above]{$1< \lambda < \mu$} ++(1.75,0);

\end{tikzpicture}
\end{center}
\caption{Range of $\lambda$ for a given $\mu$~\label{fig:lam:range}}
\end{figure}
%
%Now for each $(\mu_1, \mu_2, \lambda) \in { {\mathbb R}^+}^3$, we have
%\newcommand*{\Bigcap}{\mathbin{\scalebox{1.5}{\ensuremath{\bigcap}}}}%
%\begin{align}
%\sCapa \leq
%&\phantom{w} = \!\!\!\!\!\!\!\!\!\! \underset{(\mu_1, \mu_2, \lambda) \in { {\mathbb R}^+}^3}{\Bigcap}
%\!\!\!\!\!\!\!\!\! \{\!(R_1,R_2,D) | \mu R_1\!+\!R_2\!+\!\frac{\lambda}{2} \log \left(\frac{Q}{D}\right) \leq
% T(\mu_1, \mu_2, \lambda) \!\}.
%\end{align}
%

For a given $\mu > 1$, three regimes of $\lambda$ are of interest, as 
depicted in Figure~\ref{fig:lam:range}. 
These regimes can be identified in the $\lambda-\mu$ plane as the 
three cases marked in Figure~\ref{fig:lam:mu}.

\begin{figure}[h]
\centering
\includegraphics[scale=0.6]{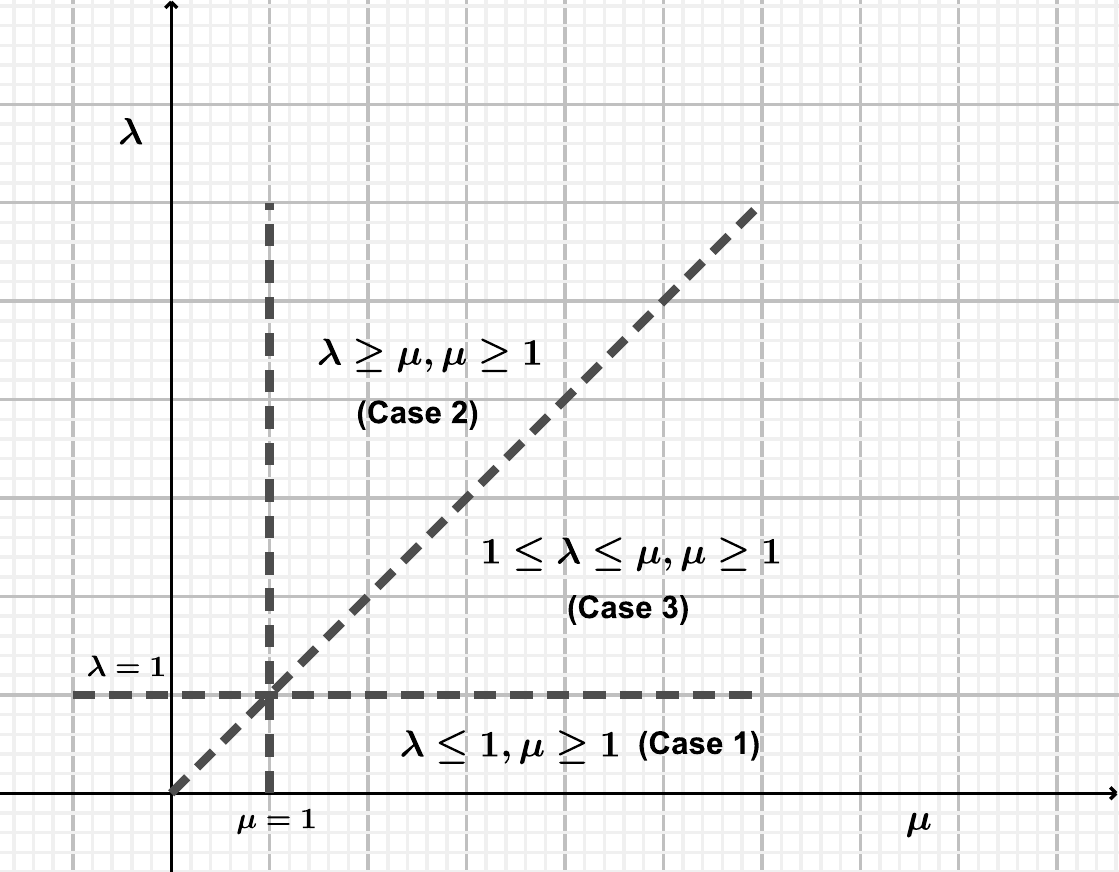}
\caption{Different Regimes involving $\lambda$ and $\mu$~\label{fig:lam:mu}}
\end{figure}

We give slightly different proofs for the three cases marked above. We begin with some 
discussion common to all the cases. Let $R_1(\gamma), R_2(\beta),
R_{sum}(\gamma,\beta)$ and $D(\gamma, \beta)$, respectively, denote the RHS of 
equations \eqref{eq:rc1} -- \eqref{eq:distc}. The following two lemmas play a key
role in our proofs for the various regimes.

\begin{lemma} \label{lem:sum:bnd}
For $\lambda \leq 1$, and $\gamma, \beta$ defined in \eqref{eq:gam:bet}, we have
\begin{multline} \label{eq:up:a}
\mu R_1 + R_2 + \frac \lambda 2 \log \frac Q{D_n} \leq \\
  (\mu-1)R_1(\gamma) + R_{sum}(\gamma, \beta) 
		+ \frac \lambda 2 \log \frac Q {D(\gamma, \beta)}+o(1).  
\end{multline}
\end{lemma}
\begin{IEEEproof}
The proof is given in Appendix~\ref{sec:app:lem:sum}.
\end{IEEEproof}

\begin{lemma} \label{lem:extreme:points}
For $\lambda>1$, the function  $f(\gamma, \beta):=R_{sum}(\gamma, \beta) 
	+ \frac 1 2 \log \frac Q {D(\gamma,\beta)}$
is a non-increasing function in each of the arguments, i.e. for $\gamma \in [0,1]$ and $\beta \in [0,1]$.
\end{lemma}

\begin{IEEEproof}
Notice that $D(\gamma, \beta)$ increases with $\gamma$ (or $\beta$), see \eqref{eq:distc}.
Furthermore, a simple inspection shows that the function $f(\gamma, \beta)$ is decreasing in 
each of the arguments. 
\end{IEEEproof}

Let us now consider the different regimes for $\lambda$ as in Figure~\ref{fig:lam:mu}.

\noindent \textbf{Case 1} ($\lambda \leq 1   \text{ and } \mu \geq 1$): In this regime, Lemma~\ref{lem:sum:bnd}
directly gives a bound on the weighted sum-rate.

\noindent \textbf{Case 2} ($\lambda \geq \mu \text{ and } \mu \geq 1$): This requires a slightly different approach than above. Since $\mu \geq 1$, we can write
\begin{align}
\mu R_1 &+ R_2 + \frac \lambda 2 \log \frac Q{D_n} \notag \\ 
	&\leq \mu R_1 +   \mu R_2 + \frac \lambda 2 \log \frac Q{D_n} \notag \\
	&= \mu (R_1 + R_2) + \frac \lambda 2 \log \frac Q{D_n} \notag \\
	&= \mu (R_1 + R_2  +  \frac 1 2 \log \frac Q{D_n}) + \frac{\lambda-\mu}2 \log \frac Q{D_n} \notag \\
	&\stackrel{(a)}{\leq} \mu \left(R_{sum}(0,0) +  \frac 12 \log \frac Q{D(0,0)}\right) + \frac {\lambda-\mu}2 \log \frac Q{D_n} \notag \\
	&\stackrel{(b)}{\leq} 0 + \frac \mu 2 \log \frac Q{D(0,0)} + \frac {\lambda-\mu}2 \log \frac Q{D(0,0)}. \label{eq:up:b}
\end{align}
In step~$(a)$ we used Lemma~\ref{lem:sum:bnd} followed by Lemma~\ref{lem:extreme:points}, and $(b)$ follows from the fact that the minimal
distortion possible is obtained by uncoded transmission of the state by the two users acting as a super-user with power $(\sqrt{P_1}+\sqrt{P_2})^2$~\cite{sutivong2005channel}. In other words,
an equivalent point-to-point channel results in the absence of messages at both the transmitters.

\noindent \textbf{Case 3} ($1\leq \lambda \leq \mu \text{ and } \mu\geq 1$): Here also we modify an
appropriate hyperplane by changing the weight on $R_2$, but this time with respect to the 
weight $\lambda$ on the distortion. More specifically, since $\lambda \geq 1$
\begin{align} 
\mu R_1 &+ R_2 + \frac \lambda 2 \log \frac Q{D_n} \notag \\ 
	&\leq \mu R_1 + \lambda R_2 +  \frac \lambda 2 \log \frac Q{D_n} \notag \\ 
	&= (\mu - \lambda) R_1 + \lambda (R_1 + R_2 + \frac 12 \log \frac Q{D_n}) \notag \\
	&\leq (\mu - \lambda) R_1 + \lambda \left(R_{sum}(\gamma,0) + \frac 12 \log \frac Q{D(\gamma,0)}\right),
	\notag
\end{align} 
where the last step used Lemmas~\ref{lem:sum:bnd} and \ref{lem:extreme:points}. From \eqref{eq:rate:alone:11}, we can infer that $R_1$ is at most $\frac 12 \log (1+\gamma P/\sigma_Z^2)$. Thus,
\begin{align} \label{eq:up:c}
\mu R_1 &+ R_2 + \frac \lambda 2 \log \frac Q{D_n} \leq \mu R_{sum}(\gamma,0) +  \frac \lambda 2 \log \frac Q{D(\gamma,0)} .
\end{align}
Let us now show that bounds \eqref{eq:up:a} -- \eqref{eq:up:c} indeed define the region given in 
Theorem~\ref{thm:main}.

\section{Equivalence of Inner and Outer Bounds}~\label{sec:equivalence}
In this section, we prove that the respective regions defined by the inner and outer bounds 
in Sections \ref{sec:achieve} and \ref{sec:conv} coincide, thereby establishing the 
capacity region. As before, for each value of the weight $\mu \geq 1$, we will consider
three different regimes for $\lambda \geq 0$, and show that the maximal value of
$\mu R_1 + R_2 + \frac \lambda 2 \log \frac QD$ in the outerbound can be achieved.
Before we embark on this, a numerical example is in order. 

Let us take  $P_1 = 2,P_2 = 2,\vZ = 1,Q = 1$.  The optimal 
trade-off is plotted in Figure~\ref{fig:3Dc}, where $(x,y,z)$-axis 
have $(R_1, R_2, \log \frac QD)$ values. In the first quadrant of $\mathbb R^3$,
$6$ distinct \emph{faces} to the region can be identified. Notice the 
oblique face along $z-$ axis in the plot. This face is a pentagon, corresponding to the
maximal distortion, however this does not coincide with $(x-y)$ plane in the plot shown (the plot begins at $z=z_m$).
%Notice that
%the $z-$ axis values inside the region do not start from zero, but 
%some positive value. 
%
This is to emphasize the fact that even if we care only about
optimizing the transmission rates, still some reduction 
of distortion from its maximal possible value of $Q$ can be achieved. 
This can also be seen from the employed DPC scheme. Essentially
the maximal distortion $D_{max} < Q$ can be achieved while operating
at the maximal sum-rate.  
In principle, the extreme pentagon at $z =  \log \frac Q{D_{max}}$  in the current plot
can be extended all the way to $z=0$.  Notice that 
%the extreme face along the $z-$ axis is in fact a pentagon. However, 
any other cross-section along $z-$axis in the interior of the plot
is not even a polytope, see Figure~\ref{fig:cross}. 

\begin{figure}[h]
\centering
\includegraphics[scale=0.6]{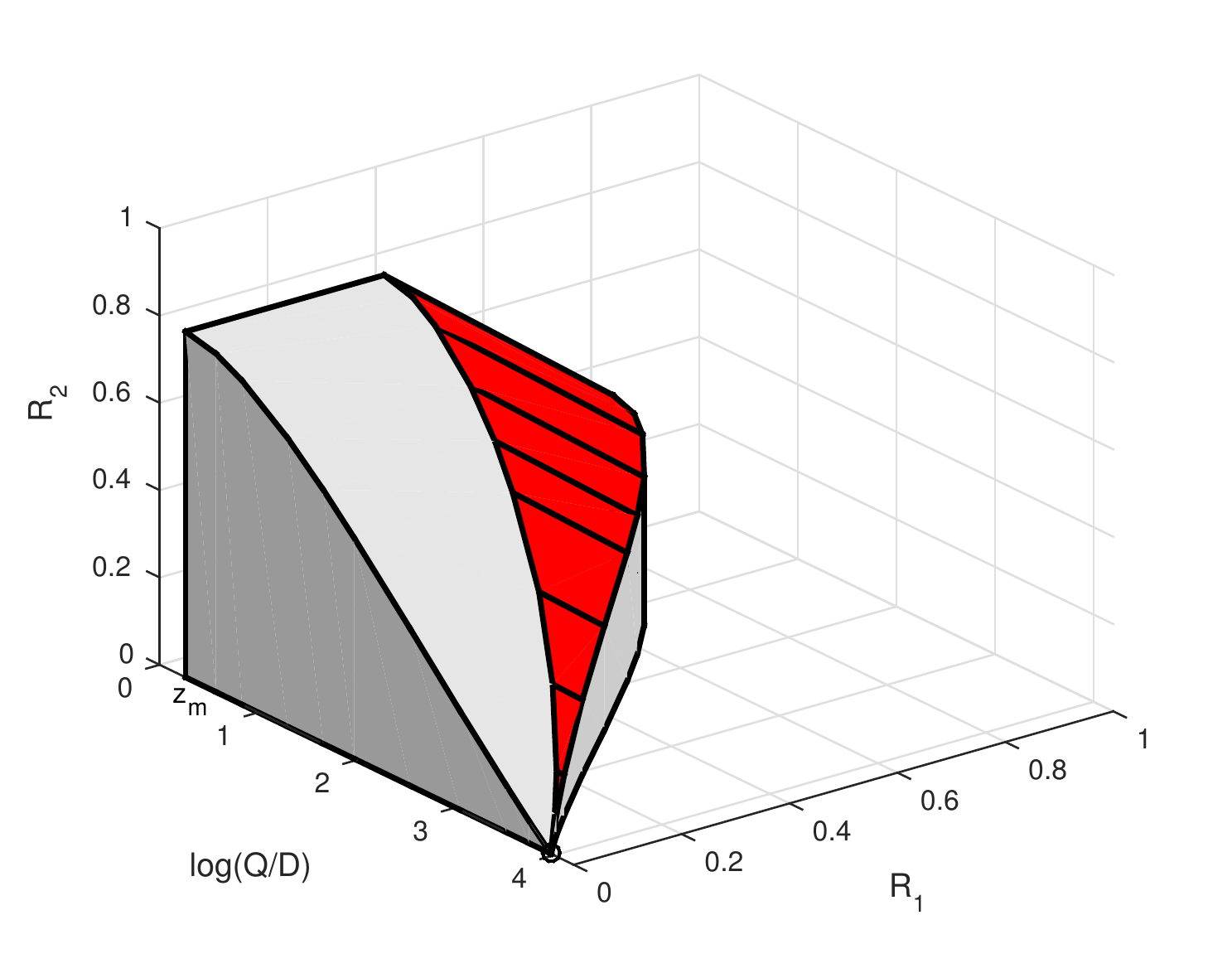}
\caption{Plot of $\sCapa$ for $P_1\!=\!2,P_2\!=\!2,\vZ\!=\!1,Q\!=\!1$.}
\label{fig:3Dc}
\end{figure}

The respective faces intersecting $x-z$ and $y-z$ axis represent the single user
tradeoff between estimation error and communication rate, when only one of
the users is present~\cite{sutivong2005channel}. Notice the three remaining faces
in the interior of $\mathbb R^3$. The middle one (striped, red) is a collection
of lines, corresponding to the sum-rate constraints at different values of distortion. 
The other two surfaces are curved, Figure~\ref{fig:cross} illustrates this
using a cross section for a given $D$ value.

\begin{figure}
\centering
\includegraphics[scale=0.6]{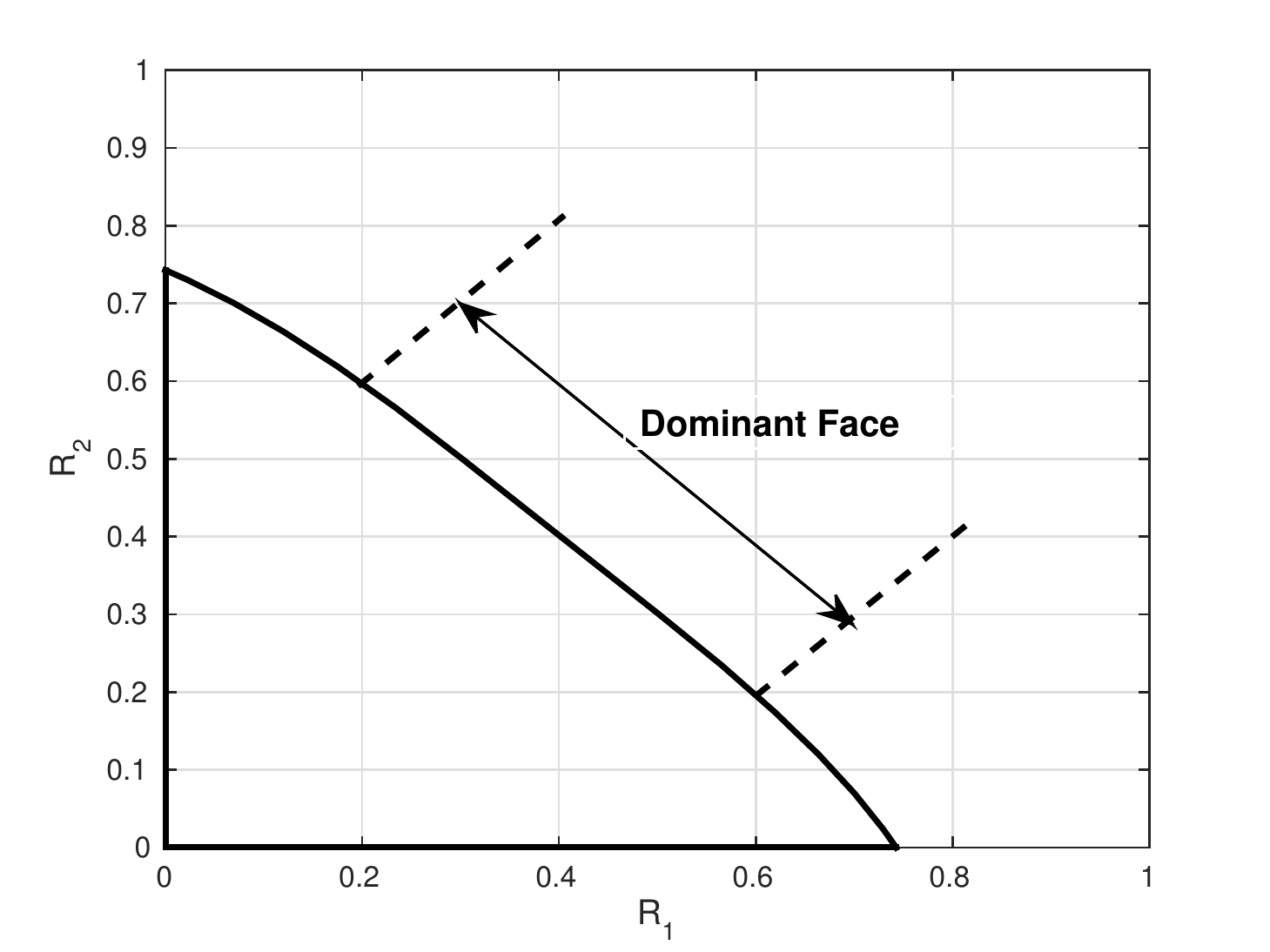}
\caption{Cross-section for $D=0.66$ in the example\label{fig:cross}}
\end{figure}

Let us now generalize our observations from the example. 
While maximizing $\mu R_1 + R_2 
+ \frac \lambda 2 \log (1+ \frac Q{D_n})$, we already showed that $\lambda \geq \mu$
corresponds to an extreme point where the sum-rate is zero (Case $2$ in Section~\ref{sec:conv}).
Clearly, the corresponding distortion lower bound $D(0,0)$ for this case
can be achieved by uncoded transmission of the state by both the transmitters, using all
the available powers.  Thus, the condition $\lambda=\mu$ subsumes all $\lambda \geq \mu$.
Furthermore the regime $1\leq \lambda \leq \mu$ (Case~$3$ of 
Section~\ref{sec:conv}) corresponds to the case where $R_2=0$. This implies that 
we need to only consider $\lambda=1$ instead of $\lambda \in [1,\mu)$. Notice 
that the region with $R_2=0$ matches the single user results of \cite{sutivong2005channel}, albeit for a state process with variance $(\sqrt{P_2}+\sqrt{Q})^2$. 
This leaves us with showing achievable schemes for those cases in which $0 < \lambda < 1$.  
The following lemma holds the key in this regime.

\begin{lemma} \label{lem:concave}
For $(0 < \lambda <1, \, \mu \geq 1)$, the function $f(\gamma,\beta) := (\mu - 1) R_1(\gamma) + R_{sum}(\gamma, \beta)
	+ \frac \lambda 2 \log \frac Q{D(\gamma,\beta)}$ is jointly strictly 
concave in $(\gamma, \beta)$ for $0\leq \gamma \leq  1$ and $0 \leq \beta \leq 1$.
\end{lemma}
\begin{IEEEproof}
The proof is relegated to Appendix~\ref{sec:app:concave}.
\end{IEEEproof}
Since we know that $\mu R_1 + R_2 + \frac \lambda 2 \log \frac Q{D_n} \leq f(\gamma, \beta)$ 
for some value of $(\gamma,\beta) \in [0,1]^2$, the strict concavity of $f(\cdot)$ suggests that
for the given $\mu >  1$ and $0 \leq \lambda \leq 1$,
there is a unique $(\gamma, \beta)$ for which $\mu R_1 + R_2 + \frac \lambda 2 \log \frac Q{D_n}$
is maximized. Clearly, choosing the maximizing parameters $\gamma, \beta$ in our achievable theorem
will give us the same operating point. Reversing the roles of $R_1$ and $R_2$, we have covered 
the whole region, except  when $\mu=1$. Anticipating the end-result, we will call the extremal surface
for $\mu=1$ as the dominant face of the plot, somewhat abusing the term \emph{face}.
Clearly $R_{sum}(\gamma, \beta) 
+ \frac \lambda 2 \log \frac Q{D(\gamma, \beta)}$ is a strictly concave function, and hence maximized
at a unique value of $(\gamma, \beta)$. Thus for a given value of distortion, the dominant line
simply connects the points $A_1$ and $A_2$ given by
\begin{align}
A_1&=\left(\frac{1}{2} \log\left(1+\frac{\gamma P_1}{\vZ}\right),\frac{1}{2} \log\left(1+\frac{\beta P_2}{\gamma P_1+\vZ}\right)\right)  \notag\\
A_2&=\left(\frac{1}{2} \log\left(1+\frac{\gamma P_1}{\beta P_2+\vZ}\right),\frac{1}{2} \log\left(1+\frac{\beta P_2}{\vZ}\right)\right)
\end{align}
for appropriate $(\gamma, \beta) \in [0,1]^2$. Evidently, each rate-pair in the dominant line is achievable by our communication scheme. This completes the proof of Theorem~\ref{thm:main}.

Let us now turn our attention towards a multi-user MAC with or without state. 
\section{Message and Source Communication for a $N-$sender GMAC}\label{sec:Nuser}
In the $N-$sender model all the transmitting nodes
observe the same source process. They should help the receiver estimate the state process. In addition
each node may have an independent stream of messages to be communicated to the base station. 
Theorem~\ref{thm:main:2} gives the optimal trade-off region. We prove this theorem below.

\subsection{Achievable Scheme}
The proof of achievability follows the same lines as before, via power sharing, dirty 
paper coding and MMSE estimation. In particular we choose
\begin{align}
X_{i \textrm{s} j} = \sqrt{\frac{(1-\gamma_i)P_i}{Q}}S_j, \:\:1 \leq i \leq N, \:\: 1 \leq j \leq n,
\end{align}
and choose $X_{ij} = X_{i \textrm{s} j} + X_{i \textrm{m} j}$, where $X_{i \textrm{m} j}\sim \mathcal N(0,\gamma_iP_i)$ 
is dirty paper coded. Suppose we employ 
successive cancellation at the decoder. Then user~$i$ does DPC to generate
 $X_{i \textrm{m}}^n$, treating $\sum_{i=1}^N X_{i \textrm{s}}^n + \alpha S^n$ as the known dirt.
The effective interference for user~$i$ is $\sum_{i \in S_i} X_{i \textrm{m}}^n + Z^n$, where $S_i$ is the
set of users decoded after user~$i$ by the successive cancellation decoder. 
Once all the messages are decoded, these signals are removed from the received symbols, and
state estimation is done using a linear MMSE estimate. Taking different user
permutations for successive cancellation, and further time-sharing 
will give the rate region given in Theorem~\ref{thm:main:2}. These straightforward
computations are omitted here.

\subsection{Converse Bound}
We now prove the converse. We are interested in maximizing
\begin{equation}
\max \sum_{i=1}^N\mu_i R_i + \frac{\lambda}{2} \log \frac{Q}{D_n},
%\max \mu_1 R_1+\mu_2 R_2 +\cdots+\mu_{N-1} R_{N-1}+R_N+\frac{\lambda}{2} \log \frac{Q}{D_n},
\end{equation}
for $\mu_i \geq 0$ and $\lambda \geq 0$. Without loss of generality, assume that $\mu_i \leq \mu_{i-1}, \forall i$.
As we did for the two user case, since $\lambda \geq 0$ is arbitrary, 
by suitable scaling, we can take $\mu_N = 1$, and $\mu_i \geq 1, 1 \leq i \leq N-1$. 

Then the following three regimes arise.
\begin{itemize}
\item $0\leq \lambda \leq 1$
\item $\lambda \geq \mu_1$
\item $\mu_j \leq \lambda \leq \mu_{j-1}$ for some $j \in \{2,\cdots,N\}$.
\end{itemize}
Let us first  extend Lemma~\ref{lem:sum:bnd} to $N-$ senders. To this end, define
for $i=1,\cdots,N$,
\begin{equation} \label{eq:rsum:gami}
R_{sum}(\gamma_1, \cdots, \gamma_i) := \frac 12 \log (1 + \sum_{j=1}^i \gamma_j P_j).
\end{equation}
Also let $D(\gamma_1, \cdots, \gamma_N)$ denote the RHS of \eqref{eq:distc:N}.
\begin{lemma} \label{lem:sum:bnd:N}
For $\lambda \leq 1$, $\mu_{N+1}=0$, and  $1=\mu_N\leq \mu_{N-1}\leq ,\cdots, \leq\mu_1$,
\begin{multline} 
\sum_{i=1}^N \mu_i R_i  +\frac{\lambda}{2} \log \frac{Q}{D_n} \\ 
%&\mu_1 R_1+\mu_2 R_2 +\cdots+\mu_{N-1} R_{N-1}+R_N+\frac{\lambda}{2} \log \frac{Q}{D_n} \notag\\
\leq  \sum_{i=2}^{N+1} (\mu_{i-1}-\mu_{i})R_{sum}(\gamma_1,\cdots,\gamma_i) \\
\phantom{ww}+\frac{\lambda}{2} \log \frac{Q}{D(\gamma_1,\cdots,\gamma_N)}+o(1). \label{eq:lemNuser}
\end{multline}
\end{lemma}
\begin{IEEEproof}
The proof is relegated to Appendix~\ref{sec:app:lem:sum:N}.
\end{IEEEproof}
The following lemma can be found true by inspection.
\begin{lemma} \label{lem:decreasing}
For $\nu \in [0,1]^N$, the function $f(\nu) \triangleq R_{sum}(\nu)+
	\frac{1}{2} \log \frac{Q}{D(\nu)}$ is a non-increasing function 
in each of the arguments $\nu_i, 1 \leq i \leq N$.
\end{lemma}
Now we consider the various regimes involving $\lambda$.

\noindent \textbf{Case 1} ($\lambda \leq 1$): In this regime, Lemma \ref{lem:sum:bnd:N} directly gives a bound on the weighted sum-rate.

\noindent \textbf{Case 2} ($\lambda \geq \mu_1$): Since $\lambda \geq \mu_i, \forall i$ in this case, we can write
\begin{align}
%&\mu_1 R_1+\mu_2 R_2 +\cdots+\mu_{N-1} R_{N-1}+R_N+
\sum_{i=1}^N \mu_i R_i &+ \frac{\lambda}{2} \log \frac{Q}{D_n} \notag\\
&\leq \mu_1 \sum_{i=1}^N R_i +\frac{\lambda}{2} \log \frac{Q}{D_n} \notag\\
&= \mu_1\left(\sum_{i=1}^N R_i +\frac{1}{2} \log \frac{Q}{D_n}\right)+
	\frac{\lambda-\mu_1}{2} \log \frac{Q}{D_n} \notag\\
&\stackrel{(a)}{\leq} \mu_1\left(R_{sum}(\bar 0)+\frac{1}{2} \log \frac{Q}{D(\bar 0)}\right) \notag\\
&\phantom{wwww}+\frac{\lambda-\mu_1}{2} \log \frac{Q}{D_n} \notag\\
&\stackrel{(b)}{\leq} 0+\frac{\mu_1}{2} \log \frac{Q}{D(\bar 0)}+\frac{\lambda-\mu_1}{2} \log \frac{Q}{D(\bar 0)},
\end{align}
where $\bar 0$ is a vector of $N$ zeros. Observe that $(a)$ is implied by 
Lemmas~\ref{lem:sum:bnd:N} -- \ref{lem:decreasing},
and $(b)$ follows since $D(\bar 0)$
is the minimal distortion possible. Notice that $D(\bar 0)$ can be achieved by
uncoded communication of the state by all transmitters.

%\noindent \textbf{Case 3} ($1 \leq \lambda \leq \mu_{N-1}$): Here we proceed as follows.
%\begin{align}
%&\mu_1 R_1+\mu_2 R_2 +\cdots+\mu_{N-1} R_{N-1}+R_N+\frac{\lambda}{2} \log \frac{Q}{D_n} \notag\\
%&\leq \mu_1 R_1+\mu_2 R_2 +\cdots+\mu_{N-1} R_{N-1} + \lambda R_N+\frac{\lambda}{2} \log \frac{Q}{D_n} \notag\\
%&= (\mu_1-\lambda) R_1+(\mu_2-\lambda)R_2+\cdots+(\mu_{N-1}-\lambda) R_{N-1} \notag\\
%&\phantom{wwww}+\lambda \left(R_{sum}+\frac{1}{2} \log \frac{Q}{D_n}\right) \notag\\
%&\leq (\mu_1-\lambda) R_1+(\mu_2-\lambda)R_2+\cdots+(\mu_{N-1}-\lambda) R_{N-1} \notag\\
%&\phantom{w}+\lambda \left(R_{sum}(\gamma_1,\gamma_2,\cdots,\gamma_{N-1},0)+\frac{1}{2} \log \frac{Q}{D(\gamma_1,\gamma_2,\cdots,\gamma_{N-1},0)}\right) \notag\\
%&\leq (\mu_1-\lambda) R_1(\gamma_1)+\cdots+(\mu_{N-1}-\lambda) R_{N-1}(\gamma_{N-1}) \notag\\
%&\phantom{w}+\lambda \left(R_{sum}(\gamma_1,\gamma_2,\cdots,\gamma_{N-1},0)+\frac{1}{2} \log \frac{Q}{D(\gamma_1,\gamma_2,\cdots,\gamma_{N-1},0)}\right).
%\end{align}

\noindent \textbf{Case 3} ($\mu_{j} \leq \lambda \leq \mu_{j-1}$):
% for some $j \in \{2,3,\cdots,N-1\}$): We take a representative case $\mu_{N-1} \leq \lambda \leq \mu_{N-2}$, and proceed as follows (others follow similarly).
Since $\lambda \geq \mu_j, \mu_{j+1},\cdots$ let us bound
\begin{align}
&\sum_{i=1}^N \mu_i R_i +\frac{\lambda}{2} \log \frac{Q}{D_n} \notag\\
&\leq \sum_{i=1}^{j-1} \mu_i R_i + \sum_{i=j}^N \lambda R_i + \frac{\lambda}{2} \log \frac{Q}{D_n} \notag\\
&=  \sum_{i=1}^{j-1} (\mu_i -\lambda) R_i + \lambda (\sum_{i=1}^N R_i + \frac{1}{2} \log \frac{Q}{D_n}) \notag\\
&\stackrel{(a)}\leq  \sum_{i=1}^{j-1} (\mu_i -\lambda) R_i \notag \\ &\phantom{ww}+ \lambda 
\Bigl(R_{sum}(\gamma_1, \cdots, \gamma_N) + \frac{1}{2} \log \frac{Q}{D(\gamma_1, \cdots, \gamma_N)}\Bigr) \notag\\
&\stackrel{(b)}\leq  \sum_{i=1}^{j-1} (\mu_i -\lambda) R_i \notag \\ &\phantom{ww}+ \lambda 
\Bigl(R_{sum}(\gamma_1, \cdots, \gamma_{j-1}) + \frac{1}{2} \log \frac{Q}{D(\gamma_1, \cdots, \gamma_{j-1},\bar 0)}\Bigr),
%&= (\mu_1-\lambda) R_1+(\mu_2-\lambda) R_2 +\cdots+(\mu_{N-2}-\lambda) R_{N-2} \notag\\
%&\phantom{www} +\lambda \left(R_{sum}+\frac{1}{2} \log \frac{Q}{D_n}\right) \notag\\
%&\leq (\mu_1-\lambda) R_1+(\mu_2-\lambda) R_2 +\cdots+(\mu_{N-2}-\lambda) R_{N-2} \notag\\
%&\phantom{ww} +\lambda \left(R_{sum}(\gamma_1,\gamma_2,\cdots,\gamma_{N-2},0,0) \right. \notag\\
%&\phantom{wwwww} \left. +\frac{1}{2} \log \frac{Q}{D(\gamma_1,\gamma_2,\cdots,\gamma_{N-2},0,0)}\right) \notag\\
%&\leq (\mu_1-\lambda) R_1(\gamma_1)+\cdots+(\mu_{N-2}-\lambda) R_{N-2}(\gamma_{N-2}) \notag\\
%&\phantom{ww} +\lambda \left(R_{sum}(\gamma_1,\gamma_2,\cdots,\gamma_{N-2},0,0) \right. \notag\\
%&\phantom{wwwww} \left.+\frac{1}{2} \log \frac{Q}{D(\gamma_1,\gamma_2,\cdots,\gamma_{N-2},0,0)}\right).
\end{align}
where $\bar 0$ is a vector of $N-j+1$ zeros, (a) follows from Lemma~\ref{lem:sum:bnd:N} with $\mu_1 = \cdots = \mu_N = \lambda$ and (b) follows from Lemma~\ref{lem:decreasing} and equation \eqref{eq:rsum:gami}. Expression \eqref{eq:txN:rate:alone} will imply that
\begin{align}
\sum_{i=1}^{j-1} \mu_i R_i \leq \sum_{i=2}^{j-1} (\mu_{i-1} - \mu_{i})  R_{sum}(\gamma_1,\cdots,\gamma_{i-1})
	\notag \\
	+ \mu_{j-1} R_{sum}(\gamma_1, \cdots, \gamma_{j-1}).
\end{align}
From the last two expressions, we get
\begin{multline} \label{eq:up:bnd:N}
\sum_{i=1}^N \mu_i R_i + \frac \lambda 2 \log \frac Q{D_n} \leq 
	 \mu_{j-1} R_{sum}(\gamma_1, \cdots, \gamma_{j-1})  \\
+ \sum_{i=2}^{j-1} (\mu_{i-1} - \mu_{i}) R_{sum}(\gamma_1, \cdots, \gamma_{i-1})  \\
+ \frac \lambda 2 \log \frac Q{D(\gamma_1, \cdots, \gamma_{j-1})} .
\end{multline}
\subsection{Equivalence of Inner and Outer bounds} \label{sec:equivalence:N}
Let us now show that the outerbound defines the same region as that can be achieved by
our communication scheme. We give a inductive argument. The base case of $N=1$ follows from \cite{sutivong2005channel}. Suppose the inner and outer
bounds given in the previous subsections are equivalent for $N-1$ users or lower. 
We will then show that the equivalence extends to $N$ users also.  Again, let us
consider different regimes for $\lambda$, $\mu_1\geq, \cdots, \geq \mu_N=1$.

For $\lambda > \mu_1$, notice that the outer bound corresponds to messages of zero rate and a 
distortion of $D(\bar 0)$.
Clearly, this can be achieved by each user sending a scaled source process, i.e. for $1 \leq k \leq N$,
\begin{align} \label{eq:pure:rate}
X_{ki} = \sqrt{\frac {P_{ki}}{Q}} S_i.
\end{align}
If $\mu_{j} \leq \lambda \leq \mu_{j-1}$, then the outerbound has terms of the form
$R_{sum}(\gamma_1, \cdots, \gamma_{l})$ with $l<j$. A natural communication choice is to 
set \eqref{eq:pure:rate} for users $j, \cdots, N$. Notice that the model now
is effectively a $j-1<N$ user problem, with a modified state process. By induction,
the inner and outerbounds coincide here. Thus, we are left with the case 
$0 \leq \lambda<1$.

Assume for simplicity that $\lambda>0$ and $\mu_i>0, 1\leq i \leq N$. Using Lemma~\ref{lem:concavityN} (Appendix~\ref{sec:app:concave}), we can
find a set of parameters $\gamma_1, \cdots, \gamma_N$, 
as the unique maximizer of the RHS in \eqref{eq:lemNuser},
for a given  $\mu_1,\cdots,\mu_N$ and $\lambda$. Clearly the achievable scheme
can get these rates by using a power split of $(1-\gamma_k)P$ and $\gamma_k P$ respectively between the
uncoded source transmission and message rate at user~$k$. This solves the $N-$ sender problem.

\section{Conclusion}\label{sec:concl}
In this paper, we considered joint message transmission and state estimation in a state dependent Gaussian multiple access channel. The optimal trade-off between the rates of the messages at two encoders and state estimation distortion was completely characterized. It was also shown that for  source and message communication over a GMAC without state, a strategy of uncoded communication at the nodes without any messages and power sharing between message transmission (using DPC) and state amplification at terminals that have a message to transmit in addition, turns out to be optimal.
% for two sender and three sender multiple access channels. Moreover, an achievable region was proposed for a general $N$ sender GMAC version of the problem. 

The discrete memoryless MAC counterpart of the current model would be an interesting open problem for further investigations.

\appendices

\section{Proof of Lemma~\ref{lem:sum:bnd}} \label{sec:app:lem:sum}
By Fano's inequality~\cite{cover2012elements}, 
we can write for any $\epsilon > 0$ for large enough $n$
\begin{equation}
H(W_1,W_2|Y^n) \leq n\epsilon.
\end{equation}
Since the $n\epsilon$ terms do not affect our end results, we will
neglect these in the sequel.

\begin{align}
& n \mu R_1+nR_2+\frac{n \lambda}{2} \log\left(\frac{Q}{D_n}\right) \notag\\
&{=} n (\mu-1) R_1+n\sum_{i=1}^2 R_i 
	+\frac{n \lambda}{2} \log\left(\frac{Q}{D_n}\right) \notag\\
&\stackrel{(a)} \leq (\mu-1)H(W_1)+H(W_1,W_2)+\lambda I(S^n;Y^n) \notag\\
&\stackrel{(b)}= (\mu-1)H(W_1|X_2^n,S^n)+H(W_1,W_2|S^n)+\lambda I(S^n;Y^n) \notag\\
&\stackrel{(c)} \approx (\mu-1)I(W_1;Y^n|X_2^n,S^n) \notag\\
&\phantom{ww}+I(W_1,W_2;Y^n|S^n)+\lambda I(S^n;Y^n) \notag\\
&= (\mu-1)I(W_1;Y^n|X_2^n,S^n)+\lambda I(W_1,W_2,S^n;Y^n) \notag\\
&\phantom{www} +(1-\lambda)I(W_1,W_2;Y^n|S^n) \notag\\
&= (\mu-1)(h(Y^n|X_2^n,S^n)-h(Y^n|W_1,X_2^n,S^n)) \notag\\
&\phantom{www} +\lambda (h(Y^n)-h(Y^n|W_1,W_2,S^n)) \notag\\
&\phantom{www}+(1-\lambda)(h(Y^n|S^n)-h(Y^n|W_1,W_2,S^n)) \notag\\
&\leq \sum_{i=1}^n \{(\mu-1)(h(Y_i|X_{2i},S_i)-h(Z_i)) \notag\\
&\phantom{www} +\lambda h(Y_i)+(1-\lambda)h(Y_i|S_i)-h(Z_i)\}, \label{eq:weightsum}
\end{align}
where $(a)$ uses Lemma~\ref{lem:dist:MI}, $(b)$ follows since $(W_1,W_2) \independent S^n$ and $W_1 \independent X_2^n$, and  $(c)$  follows from Fano's inequality. Let us now upper bound the term 
$\lambda h(Y_i)+(1-\lambda)h(Y_i|S_i)$ in \eqref{eq:weightsum}. 
Notice that $\lambda h(Y_i)+(1-\lambda)h(Y_i|S_i),\:\: \lambda \in [0,1]$ is simultaneously maximized (for fixed $K_i(X_{1i},X_{2i},S_i)$) when $(X_{1i}+X_{2i})$ is jointly Gaussian with $S_i$~\cite{sutivong2005channel}.
Without loss of generality, for the purposes of finding an upper bound on $\mu R_1+R_2+\frac{\lambda}{2} \log \frac{Q}{D_n}$, we can express, using \eqref{eq:ki:entry}
\begin{align*}
\sum_{j=1}^2X_{ji} =V_i+\left(\eta_{1i} \sqrt{(1-\gamma_i)\frac{P_{1i}}Q} 
	+ \eta_{2i}\sqrt{(1-\beta_i)\frac{P_{2i}}Q}
	\right)S_i,
\end{align*}
where  $V_i$ is  zero mean Gaussian, independent of  $S_i$.
The second term on the RHS can be understood as the linear estimate of $(X_{1i}+X_{2i})$ 
given $S_i$. 
Since $V_i$ and $S^n$ are 
independent,
\begin{align*}
\Var[V_i] &= \Var (X_{1i}+ X_{2i}|S^n) \\
	  &\leq \Var (X_{1i}|S_i) + \Var(X_{2i}|S_{i}) \\
	  &= \gamma_i P_{1i} + \beta_i P_{2i}.
\end{align*}
%
%\begin{align}
%(X_{1i}+X_{2i})=V_i+(\sqrt{(1-\gamma_i)P_{1i}/Q}+\sqrt{(1-\beta_i)P_{2i}/Q})S_i, \label{eq:jointgauss}
%\end{align}
%where $V_i$ is Gaussian and independent of $S_i$. 
%Indeed, we have by joint Gaussianity
%\begin{align}
%\eE{[(X_{1i}+X_{2i})|S_i]} = (\eE{[(X_{1i}+X_{2i})S_i]}/Q) \: S_i.
%\end{align} 
%Now we shall bound the variance of the random variable $V_i$ in \eqref{eq:jointgauss}. On representing \eqref{eq:jointgauss} in vector notation
%\begin{align}
%X_1^n+X_2^n = V^n+A S^n.
%\end{align}
%where $A$ is a suitable diagonal matrix. Notice that
%\begin{align*}
%\Cov{[X_1^n+X_2^n|S^n]}=\Cov{[V^n+A S^n|S^n]}=\Cov{[V^n]}.
%\end{align*}
%Now using the Markov condition $X_1^n \to S^n \to X_2^n$, we have
%\begin{align*}
%\Cov{[V^n]}\!=\!\Cov{[X_1^n\!+\!X_2^n|S^n]}=\Cov{[X_1^n|S^n]}\!+\!\Cov{[X_2^n|S^n]}.
%\end{align*}
%For the $i$-th diagonal entry, the above can be written as
%\begin{align}
%\Var{[V_i]} &= \Var{[X_{1i}|S^n]}+\Var{[X_{2i}|S^n]} \notag\\
%&\leq \sigma_{X_{1i}|S_i}^2+\sigma_{X_{2i}|S_i}^2 = \gamma_i P_{1i}+\beta_i P_{2i}.
%\end{align}
%where the inequality follows because conditional variances are smaller than unconditional variances on the average. Hence to get an upper bound on \eqref{eq:weightsum}, we can take
%\begin{align*}
%(X_{1i}+X_{2i})=V_i+(\sqrt{(1-\gamma_i)P_{1i}/Q}+\sqrt{(1-\beta_i)P_{2i}/Q})S_i,
%\end{align*}
%where $V_i \sim \mathcal{N}(0,\gamma_i P_{1i}+\beta_i P_{2i})$ and $V_i \independent S_i$. Notice that
Using this
\begin{align*}
\eE{[(X_{1i}+X_{2i})]^2} \leq P_{1i}+P_{2i}+2\sqrt{(1-\gamma_i)(1-\beta_i)P_{1i}P_{2i}},
\end{align*}
where we have taken $\eta_{1i}=\eta_{2i}=1$ as the sign of correlation in \eqref{eq:ki:entry},
as negative correlation can only be detrimental for the RHS.
Now on denoting $g(x)=(1/2)\log(2{\pi}e x)$ and using the differential entropy maximizing property of Gaussian random variables for a given variance, we can write
\begin{align}
%h(Y_i|X_{1i},S_i)  &\leq h(X_{2i}+Z_i|S_i) \notag\\
%&\leq g(\sigma_{X_{2i}|S_i}^2+\vZ) =  g(\beta_i P_{2i}+\vZ), \label{eq:tx1:rate}\\
h(Y_i|X_{2i},S_i) &\leq h(X_{1i}+Z_i|S_i) \notag\\
&\leq g(\Var{[X_{1i}|S_i]}+\vZ) = g(\gamma_i P_{1i}+\vZ) \label{eq:tx1:rate}.
\end{align}
%Now using the differential entropy maximizing property of Gaussian random variables for a given variance, we can write
\begin{align}
h(Y_i|S_i) &= h(V_i+Z_i) \leq g(\gamma_i P_{1i}+\beta_i P_{2i}+\vZ), \label{entr3}\\
h(Y_i) &\leq g(Q\!+\!\vZ\!+\!\eE{[X_{1i}+\!X_{2i}]^2}\!+\!2\eE{[(X_{1i}\!+\!X_{2i})S_i]}) \notag\\
&\leq g(P_{1i}+P_{2i}+Q+\vZ \notag\\
&\phantom{www}+2\sqrt{\bar{\gamma}_i P_{1i}Q} +2\sqrt{\bar{\beta}_i P_{2i}Q}+2\sqrt{\bar{\gamma}_i \bar{\beta}_i P_{1i}P_{2i}}) \label{entr4},
\end{align}
where \eqref{entr3} is under the choice of $X_{1i}+X_{2i}$ which maximizes $\lambda h(Y_i)+(1-\lambda) h(Y_i|S_i)$, for all $\lambda \in [0,1]$.
Continuing the chain of inequalities from \eqref{eq:weightsum}:
\begin{align}
& n \mu R_1+nR_2+\frac{n \lambda}{2} \log\left(\frac{Q}{D_n}\right) \notag\\
& \stackrel{(a)}\leq \sum_{i=1}^n (\mu-1) \frac{1}{2} \log  \left(\frac{\gamma_i P_{1i}+\vZ}{\vZ}\right) \notag\\
&\phantom{w} + \sum_{i=1}^n \frac{\lambda}{2} \log \! \left(\!\!\left(
\begin{aligned}
&P_{1i}+P_{2i}+Q+\vZ+2\sqrt{\bar{\gamma_i}P_{1i} Q}\\
&+2\sqrt{\bar{\beta_i}P_{2i} Q}+2\sqrt{\bar{\gamma_i}\bar{\beta_i}P_{1i} P_{2i}}
\end{aligned}
\!\right) \!\!
\Bigg / \!\! \vZ\right) \notag\\
&\phantom{w} +\sum_{i=1}^n \frac{(1-\lambda)}{2} \log\left(\frac{\gamma_i P_{1i}+\beta_i P_{2i}+\vZ}{\vZ}\right) \notag\\
&\stackrel{(b)} \leq (\mu-1) \frac{n}{2} \log\left(\frac{\gamma P_1+\vZ}{\vZ}\right) \notag\\
&\phantom{w} +\frac{\lambda n}{2}\log\left(\left(
\begin{aligned}
&P_{1}+P_{2}+Q+\vZ+2\sqrt{\bar{\gamma}P_{1} Q}\\
&+2\sqrt{\bar{\beta}P_{2} Q}+2\sqrt{\bar{\gamma}\bar{\beta}P_{1} P_{2}}
\end{aligned}
\right)
\Bigg/ \vZ \right) \notag\\
&\phantom{w} +\frac{n(1-\lambda)}{2} \log\left(\frac{\gamma P_{1}+\beta P_{2}+\vZ}{\vZ}\right) \notag\\
&= (\mu-1) \frac{n}{2} \log\left(\frac{\gamma P_1+\vZ}{\vZ}\right) \notag\\
&\phantom{w} + \frac{\lambda n}{2}\log \! \left(\!\!\left(
\begin{aligned}
&P_{1}\!+\!P_{2}\!+\!Q+\!\vZ\!+\!2\sqrt{\bar{\gamma}P_{1} Q}\\
&\!+\!2\sqrt{\bar{\beta}P_{2} Q}\!+\!2\sqrt{\bar{\gamma}\bar{\beta}P_{1} P_{2}}
\end{aligned}
\right) \!\!
\Bigg/ \!\! \left(\!
\begin{aligned}
&\vZ\!+\!\gamma P_1 \\
&\phantom{w}\!+\!\beta P_2
\end{aligned}
\!\right) \!\!\right) \notag\\
&\phantom{w} +\frac{n}{2} \log\left(\frac{\gamma P_{1}+\beta P_{2}+\vZ}{\vZ}\right) \notag\\
&\stackrel{(c)} = (\mu-1) \frac{n}{2} \log\left(\frac{\gamma P_1+\vZ}{\vZ}\right) \notag\\
&\phantom{ww}+ nR_{sum}(\gamma,\beta)+\frac{\lambda n}{2} \log \left(\frac{Q}{D(\gamma,\beta)}\right). \label{convfs}
\end{align}
where (a) follows from the fact that both $\lambda$ and $(1-\lambda)$ are non-negative for $\lambda \in [0,1]$ and expression \eqref{entr3} and \eqref{entr4}, (b) follows from Jensen's Inequality, and (c) follows from the definitions of $R_{sum}(\gamma,\beta)$, $D(\gamma,\beta)$ from \eqref{eq:rc3} and \eqref{eq:distc}. Thus 
the lemma is proved for all $\lambda \in [0,1]$. 
%By taking $\lambda \rightarrow 0$, we also get 
From \eqref{eq:tx1:rate}, we also get
\begin{align} \label{eq:rate:alone:11}
R_1 \leq \frac 12 \log(1+ \gamma P_1).
\end{align}

%For the parameters mentioned above, let $\gamma,\beta \geq 0$ be such that
%\begin{align}
%\gamma P_1 &= \frac{1}{n} \sum_{i=1}^n \gamma_i P_{1i} \:\: \textup{and} \:\: \beta P_2 = \frac{1}{n} \sum_{i=1}^n \beta_i P_{2i}. \label{eq:param}
%\end{align}
%Since $nP_1 \geq \sum_{i=1}^n P_{1i}$ and $\gamma_i \leq 1$, $\gamma \leq 1$. Likewise, $\beta \leq 1$. 

\section{Proof of Lemma \ref{lem:concave}} \label{sec:app:concave}
Here we prove a slightly more general result, which turns out useful for the
remaining sections. Consider a  concave function $L(\nu),\nu \in [0,1]^N$, and let us define
\begin{align} \label{eq:obj:conc}
f(\nu) := \sum_{i=1}^N \alpha_i \log(1+ \sum_{j=1}^i \nu_j P_j) + \frac \lambda 2 \log L(\nu),
\end{align}
where $\alpha_i, 1\leq i \leq N$ and $\lambda$ are non-negative constants.
\begin{lemma} \label{lem:concavityN}
For $0<\lambda \leq 1$, $f(\cdot)$ is strictly concave in $\nu
\in (0,1)^N$, whenever $\alpha_i, 1 \leq i\leq N$ are not identically zero.
\end{lemma}
\begin{IEEEproof}
The first term, being the linear combination of logarithms, is strictly concave. Let us consider the second term. Let $x_1$ and $x_2$ be two $N-$ dimensional vectors in $\mathbb R^N$. Notice that for $\zeta \in [0,1]$,
\begin{align}
\zeta \log L(x_1) + (1\!-\!\zeta) \log L(x_2) 
&\leq \log(\zeta L(x_1) \!+\! (1\!-\!\zeta)L(x_2)) \notag \\
&\leq \log L(\zeta x_1 + \zeta x_2),
\end{align}
since $L(\cdot)$ itself is concave by assumption.  This proves the lemma.
\end{IEEEproof}
Let us proceed to show Lemma~\ref{lem:concave}. Denote
\begin{align} \label{eq:L:nu}
L(\gamma,\beta) &= P_1+P_2+Q+\vZ+2\sqrt{\bar{\gamma} P_1 Q} \notag\\
&\phantom{w}+2\sqrt{\bar{\beta} P_2 Q}+2\sqrt{\bar{\gamma} \bar{\beta} P_1 P_2},
\end{align}
for convenience. 
%We need to prove the strict joint concavity of the function $f(\gamma,\beta)=(\mu-1)\frac{1}{2} \log \left(1+\frac{\gamma P_1}{\vZ}\right)+R_{sum}(\gamma,\beta)+\frac{\lambda}{2}\log\left(\frac{Q}{D(\gamma,\beta)}\right)$ in $\gamma,\beta \in [0,1]$ for $\lambda \in [0,1], \mu \geq 1$. Notice that
%Since $\lambda \in [0,1]$, both $\lambda$ and $1-\lambda$ are non-negative. Hence we rewrite $f(\gamma,\beta)$ as follows
Rewriting the function given in lemma, 
\begin{align}
f(\gamma,\beta) &= \frac{(\mu-1)}{2} \log \left(1+\frac{\gamma P_1}{\vZ}\right) \notag\\
&\phantom{w}+\frac{(1-\lambda)}{2} \log \left(\!1\!+\!\frac{\gamma P_1\!+\!\beta P_2}{\vZ}\!\right)+\frac{\lambda}{2} \log \left(\!\frac{L(\gamma,\beta)}{\vZ}\!\right).
\end{align}
Thus $f(\gamma,\beta)$ is a sum similar to \eqref{eq:obj:conc}. Our proof will be complete by showing $L(\cdot)$ in \eqref{eq:L:nu} to be a concave function.
For $c_0>0$, and non-negative constants $c_1, \cdots, c_N$, the function
$$
L(\nu) = c_0 + \sum_{i=1}^N \sum_{j=1}^N 
c_i c_j \sqrt {(1-\nu_i)(1-\nu_j})
$$
is concave. To see this, notice that $\sqrt{x}$ is strictly concave in $x\geq 0$. 
Also, $\sqrt{xy}$ is jointly concave in $(x,y) \in [0,1]^2$, making $L(\nu)$ a concave function. 
Notice that concavity in the range of interest is maintained by replacing each and every
 variable  $x\in[0,1]$ by $1-x$.
The proof of the lemma is now complete.

\section{Proof of Lemma \ref{lem:sum:bnd:N}}~\label{sec:app:lem:sum:N}
%We use the following shorthand notations.
%\begin{align}
%W_{[k]} &= (W_1,W_2,\cdots,W_k), \\
%X_{[N] \setminus j}^n &= (X_1^n,X_2^n,\cdots,X_{j-1}^n,X_{j+1}^n,\cdots,X_N^n).
%\end{align}
Recall that $\mu_1\geq \mu_2\geq, \cdots, \geq \mu_N=1$. We need to show \eqref{eq:lemNuser} for $\lambda \leq 1$. Let us define $\mu_{N+1}:=0$. 
%For economy of space, we also take $\alpha=0$. 
Now
\begin{align}
& n \sum_{j=1}^{N} \mu_j R_j +\frac{n \lambda}{2} \log\left(\frac{Q}{D_n}\right) \notag\\
&= n \sum_{j=2}^{N+1} (\mu_{j-1}-\mu_j) \sum_{i=1}^{j-1} R_i  +\frac{n \lambda}{2} \log\left(\frac{Q}{D_n}\right) \notag\\
&=  \sum_{j=2}^{N+1} (\mu_{j-1}-\mu_j) H(W^{j-1})  +\frac{n \lambda}{2} \log\left(\frac{Q}{D_n}\right) \notag\\
&=  \sum_{j=2}^{N} (\mu_{j-1}-\mu_j) H(W^{j-1})  + H(W^N) + \frac{n \lambda}{2} \log\left(\frac{Q}{D_n}\right) \notag\\
&\stackrel{(a)}{\leq}  \sum_{j=2}^{N} (\mu_{j-1}-\mu_j) I(W^{j-1};Y^n|S^n, X_{j}^n,\cdots, X_{N}^n) \notag \\
&\phantom{ww} + I(W^N;Y^n|S^n) + \frac{n \lambda}{2} \log\left(\frac{Q}{D_n}\right) \notag\\
&\leq  \sum_{j=2}^{N} (\mu_{j-1}-\mu_j) I(W^{j-1};Y^n|S^n, X_{j}^n,\cdots, X_{N}^n) \notag \\
&\phantom{ww} + I(W^N;Y^n|S^n) + \lambda I(S^n;Y^n)\notag\\
&= \sum_{j=2}^{N} (\mu_{j-1}-\mu_j) I(W^{j-1};Y^n|S^n, X_{j}^n,\cdots, X_{N}^n) \notag \\
&\phantom{ww} + (1-\lambda) I(W^N;Y^n|S^n) + \lambda I(S^n W^N;Y^n)\notag\\
&= \sum_{j=2}^{N} (\mu_{j-1}-\mu_j) (h(Y^n|S^n, X_{j}^n,\cdots, X_{N}^n) - h(Z^n)) \notag \\
&\phantom{ww} + (1-\lambda) \left(h(Y^n|S^n) - h(Z^n)\right) + \lambda \left(h(Y^n) - h(Z^n)\right)\notag\\
&= \sum_{j=2}^{N} (\mu_{j-1}-\mu_j) \sum_{i=1}^n (h(Y_i|S^n , X_{ji},\cdots, X_{Ni}) - h(Z_i)) \notag \\
&\phantom{i} + (1-\lambda) \sum_{i=1}^nh(Y_i|S_i)  + \lambda \sum_{i=1}^n h(Y_i) - \sum_{i=1}^n h(Z_i)
\label{eq:N:up:1}
%
%& \leq \sum_{j=1}^{N-1} (\mu_j-1)H(W_j)+H(W_{[N]})+\lambda I(S^n;Y^n) \notag\\
%&\stackrel{(a)}= \sum_{j=1}^{N-1} (\mu_j-1) H(W_j|X_{[N] \setminus j}^n,S^n) \notag\\
%&\phantom{www}+H(W_{[N]}|S^n)+\lambda I(S^n;Y^n) \notag\\
%&\stackrel{(b)} \approx \sum_{j=1}^{N-1} (\mu_j-1) I(W_j;Y^n|X_{[N] \setminus j}^n,S^n) \notag\\
%&\phantom{www}+I(W_{[N]};Y^n|S^n)+\lambda I(S^n;Y^n) \notag\\
%&= \sum_{j=1}^{N-1} (\mu_j-1) I(W_j;Y^n|X_{[N] \setminus j}^n,S^n)+\lambda I(W_{[N]},S^n;Y^n) \notag\\
%&\phantom{www} +(1-\lambda)I(W_{[N]};Y^n|S^n) \notag\\
%&= \sum_{j=1}^{N-1} (\mu_j-1) (h(Y^n|X_{[N] \setminus j}^n,S^n)-h(Y^n|W_j,X_{[N] \setminus j}^n,S^n)) \notag\\
%&\phantom{www} +\lambda (h(Y^n)-h(Y^n|W_{[N]},S^n)) \notag\\
%&\phantom{www}+(1-\lambda)(h(Y^n|S^n)-h(Y^n|W_{[N]},S^n)) \notag\\
%&\leq \sum_{j=1}^{N-1} (\mu_j-1) \left(\sum_{i=1}^n (h(Y_i|X_{{[N] \setminus j}_i},S_i)-h(Z_i))\right) \notag\\
%&\phantom{www} +\sum_{i=1}^n (\lambda h(Y_i)+(1-\lambda)h(Y_i|S_i)-h(Z_i)), \label{eq:weightsumN}
\end{align}
where (a) follows since $W^N \independent S^n$, $W^{j-1} \independent (S^n,X_j^n,\cdots, X_N^n)$ and Fano's inequality. 

Let us now consider the covariance matrix $K_i$ of $(X_{1i}, \cdots, X_{Ni}, S_i)$. For some 
$\gamma_{ki} \in [0,1]$ we can write
\begin{align}
\eE X_{ki}S_i = \eta_{ki} \sqrt{(1-\gamma_{ki}) P_{ki}Q},
\end{align}
where $P_{ki} = K_i(k,k), 1 \leq k \leq N$ is the empirical average power of user~$k$ 
for transmission index $i$ in a block, and $\eta_{ki} \in \{-1,+1\}$ is the sign of
the correlation. From this, we also get
$$
\Var(X_{ki}|S_i) = \gamma_{ki} P_{ki}.
$$
Let us  now bound the entropy terms in \eqref{eq:N:up:1}. 
\begin{align}
h(Y_i|S^n, X_{ji}, \cdots, X_{Ni}) &\leq g(\Var(\sum_{k=1}^{j-1}X_{ki}+ Z_i|S^n)) \notag \\
	&\leq  g(\sum_{k=1}^{j-1} \Var(X_{ki}+ Z_i|S^n)) \notag \\
	&\leq  g(\sum_{k=1}^{j-1} \Var(X_{ki}|S_i) + \vZ) \notag \\
	&\leq g(\vZ + \sum_{k=1}^{j-1} \gamma_{ki} P_{ki}), \label{eq:txN:rate}
\end{align}
where the first two expressions used the Markov condition 
$X_{ki}\rightarrow S^n \rightarrow X_{ji}$. Observe that $\lambda h(Y_i)+(1-\lambda)h(Y_i|S_i),\:\: \lambda \in [0,1]$ would be maximized (for fixed covariance $K_i$) when $(X_{1i}+\cdots+X_{Ni})$ is jointly Gaussian 
with $S_i$. 
%Similar to \cite{sutivong2005channel}, we shall represent (refer to Lemma $1$ in \cite{sutivong2005channel} for the details)
Let us write
\begin{align}
\sum_{j=1}^n X_{j_i} = V_i+ \left(\sum_{j=1}^N \eta_{ji} \sqrt{(1-\gamma_{j_i})P_{j_i}/Q} \right)S_i, \label{eq:jointgaussN},
\end{align}
%where $V_i$ is Gaussian and independent of $S_i$.
%Now we shall bound the variance of the random variable $V_i$ in \eqref{eq:jointgaussN}. On representing \eqref{eq:jointgaussN} in vector notation
%\begin{align}
%\sum_{j=1}^N X_j^n = V^n+A S^n.
%\end{align}
%where $A$ is a suitable diagonal matrix. Notice that
%\begin{align*}
%\Cov{[X_1^n+\cdots+X_N^n|S^n]}=\Cov{[V^n+A S^n|S^n]}=\Cov{[V^n]}.
%\end{align*}
%Since the inputs are independent of each other given $S^n$, we have
%\begin{align*}
%\Cov{[V^n]} = \sum_{j=1}^N \Cov{[X_j^n|S^n]}.
%\end{align*}
%For the $i$-th diagonal entry, the above can be written as
%\begin{align}
%\Var{[V_i]} &= \sum_{j=1}^N \Var{[X_{j_i}|S^n]} \leq \sum_{j=1}^N \sigma_{X_{j_i}|S_i}^2.
%\end{align}
%where the inequality follows because conditional variances are smaller than unconditional variances on the average. Hence to get an upper bound on \eqref{eq:weightsumN}, we can take
%\begin{align*}
%\sum_{j=1}^n X_{j_i} =V_i+ \left(\sum_{j=1}^N \sqrt{(1-\gamma_{j_i})P_{j_i}/Q} \right)S_i,
%\end{align*}
where $V_i$ is a zero mean Gaussian with  $\Var V_i \leq \sum_{j=1}^N \gamma_{ji} P_{ji}$, 
and $V_i \independent  S_i$. Since $\forall j, \, |\eta_{ji}|=1$, using  
\begin{align*}
\eE{\left[\sum_{j=1}^n X_{j_i}\right]^2} \!\!\!\! \leq \sum_{j=1}^N \!\! P_{ji}\!+2 \sum_{j=1}^N \sum_{k \neq j} \sqrt{(1-\gamma_{j_i})(1-\gamma_{k_i})P_{ji}P_{ki}}.
\end{align*}
%Now using the differential entropy maximizing property of Gaussian random variables for a given variance, we can write
Putting these altogether
\begin{align}
h(Y_i|S_i) &= h(V_i+Z_i) \leq g\left(\sum_{j=1}^N \gamma_{ji} P_{ji} +\vZ\right), \label{eq:entr3N}\\
h(Y_i) &\leq g\left(\sum_{j=1}^N P_{ji}+\alpha^2 Q+\vZ+2 \alpha \sum_{j=1}^N \sqrt{\bar{\gamma}_{ji} P_{ji} Q} \right.\notag\\
&\phantom{wwwww} \left.  + 2 \sum_{j=1}^N \sum_{k \neq j} \sqrt{\bar{\gamma}_{ji} \bar{\gamma_k}_i P_{ji}P_{ki}} \right) \label{eq:entr4N}.
\end{align}
Defining $\gamma_k = \frac 1{nP_k} \sum_{i=1}^n \gamma_{ki} P_{ki}$, 
for $\lambda \in (0,1)$, we get from \eqref{eq:N:up:1}, \eqref{eq:txN:rate} -- \eqref{eq:entr4N}  
\begin{align}
  \sum_{j=1}^{N} \mu_j R_j &+\frac{ \lambda}{2} \log\left(\frac{Q}{D_n}\right) \notag\\
&\leq  \sum_{j=2}^{N} (\mu_{j-1}-\mu_j) g(\vZ + \sum_{k=1}^{j-1} \gamma_k P_k) \notag \\
&\phantom{ww} + (1-\lambda)  g(\vZ + \sum_{j=1}^N \gamma_j P_j )   \notag \\
&\phantom{ww}+ \lambda g\left(\sum_{j=1}^N\!P_{j}+\alpha^2 Q+\vZ+2\alpha \sum_{j=1}^N\!\! \sqrt{\bar{\gamma_j} P_{j} Q} \right.\notag\\
&\phantom{www} \left.  + 2 \sum_{j=1}^N \sum_{k \neq j} \sqrt{\bar{\gamma_j} \bar{\gamma_k} P_{j}P_{k}} \right) - \mu_1 g(\vZ).
\end{align}
Here we used Jensen's inequality on \eqref{eq:txN:rate}, \eqref{eq:entr3N} and \eqref{eq:entr4N} 
to get a single letter form independent of the transmission index.
This proves Lemma~\ref{lem:sum:bnd:N}. Note that for $S \subseteq \{1,\cdots,N\}$, by giving $(S^n,X_{j \in \mathcal{S}^c}^n)$ to the receiver and using Jensen's inequality, we obtain (by following similar steps as \eqref{eq:txN:rate}),
\begin{align} \label{eq:txN:rate:alone}
n \sum_{j\in S} R_j &\leq I(X_{j \in \mathcal{S}}^n;Y^n|S^n,X_{j \in \mathcal{S}^c}^n) \notag\\
&= h(Y^n|S^n,X_{j \in \mathcal{S}^c}^n)-h(Z^n) \notag\\
&\leq \frac {n}{2} \log (1 + \frac{\sum_{j\in S} \gamma_j P_j}{\vZ}), \forall S \subseteq \{1,\cdots,N\}.
\end{align}

\bibliographystyle{IEEEtran}
\bibliography{mybib}
%\nocite{}

% You can push biographies down or up by placing
% a \vfill before or after them. The appropriate
% use of \vfill depends on what kind of text is
% on the last page and whether or not the columns
% are being equalized.

%\vfill

% Can be used to pull up biographies so that the bottom of the last one
% is flush with the other column.
%\enlargethispage{-5in}

% that's all folks
\end{document}